\documentclass[sn-mathphys]{sn-jnl}

\usepackage{comment}
\usepackage[utf8]{inputenc}
\usepackage{float}
\usepackage{alphabeta}
\usepackage{indentfirst}
\usepackage{mathtools}

\usepackage{amsmath}
\usepackage{amssymb}

\usepackage{setspace}
\RequirePackage{fix-cm}
\usepackage{titlesec}
\usepackage{physics}
\usepackage{array}
\usepackage{nameref}
\usepackage{cleveref}
\usepackage{dsfont}
\usepackage{xcolor}
\usepackage{enumitem}
\usepackage{thmtools}

\Crefname{theorem}{Theorem}{Theorems}
\Crefname{lem}{Lemma}{Lemmata}
\Crefname{defn}{Definition}{Definitions}
\Crefname{equation}{Eq.}{Equations}

\theoremstyle{thmstyleone}
\newtheorem{theorem}{Theorem}[section]
\newtheorem{prop}[theorem]{Proposition}
\newtheorem{lem}[theorem]{Lemma}
\newtheorem{corollary}[theorem]{Corollary}
\newtheorem{defn}[theorem]{Definition}
\newtheorem{remark}[theorem]{Remark}
\newtheorem*{remark*}{Remark}

\DeclareMathOperator{\dist}{dist}
\DeclareMathOperator{\supp}{supp}
\DeclareMathOperator{\diam}{diam}

\DeclarePairedDelimiter\floor{\lfloor}{\rfloor}

\newcommand{\Srat}{\mathbb{S}_{\mathbb Q}}
\newcommand{\Z}{{\mathbb{Z}}}
\newcommand{\N}{{\mathbb{N}}}
\newcommand{\R}{{\mathbb{R}}}
\newcommand{\C}{{\mathbb{C}}}

\begin{document}

\title[Almost everywhere ergodicity in quantum lattice models]{Almost everywhere ergodicity in quantum lattice models}
\author*[1]{\fnm{Dimitrios} \sur{Ampelogiannis}}\email{dimitrios.ampelogiannis@kcl.ac.uk}
\author*[1]{\fnm{Benjamin} \sur{Doyon}}\email{benjamin.doyon@kcl.ac.uk}

\affil[1]{\orgdiv{Department of Mathematics}, \orgname{ King's College London}, \orgaddress{\street{Strand WC2R 2LS}, \city{UK}}}

\abstract{
%Obtaining rigorous and general results about the quantum dynamics of extended many-body systems is a difficult task. Given the panoply of phenomena that can be observed and are being investigated, it is crucial to delineate the boundaries of what is possible. In quantum lattice models, the Lieb-Robinson bound tells us that the spatial extent of operators grows at most linearly in time. But what happens within this light-cone? We obtain a universal form of ergodicity showing that operators get ``thinner'' almost everywhere within the light-cone. We show how this leads to a universal hydrodynamic projection formula for the large-time behaviour of correlation functions, defining accurately the space of emergent degrees of freedom. The results are general, applicable to any locally interacting system, at arbitrary frequency and wavenumber.
We rigorously examine, in generality, the ergodic properties of quantum lattice models with short range interactions, in the $C^*$ algebra formulation of statistical mechanics. Ergodicity results, in the context of group actions on $C^*$ algebras, assume that the algebra is asymptotically abelian, which is not generically the case for time evolution. The Lieb-Robinson bound tells us that, in a precise sense, the spatial extent of any time-evolved local operator grows linearly with time. This means that the algebra of observables is asymptotically abelian in a space-like region, and implies a form of ergodicity outside the light-cone. But what happens within it? We show that the long-time limit of the $n$-th moment of a ray-averaged observable, along space-time rays of almost every speed, converges to the $n$-th power of its expectation in the state (i.e.\ its ensemble average). Thus ray averages do not fluctuate in the long time limit. This is a statement of ergodicity, and holds in any state that is invariant under space-time translations and that satisfies weak clustering properties in space. The ray averages can be performed in a way that accounts for oscillations, showing that ray-averaged observables cannot sustainably oscillate in the long time limit. We also show that in the GNS representation of the algebra of observables, for any KMS state with the above properties, the long-time limit of the ray average of any observable converges (in the strong operator topology) to the ensemble average times the identity, again along space-time rays of almost every speed. This is a strong version of ergodicity, and indicates that, as operators, observables get ``thinner'' almost everywhere within the light-cone. A similar statement holds under oscillatory averaging.
}

\maketitle

%\setcounter{tocdepth}{2}
%\tableofcontents

\section{Introduction}
In this paper we are concerned with the ergodic properties of quantum spin lattice systems of arbitrary spatial dimension $D$, with exponentially decaying (including finite range) interactions. Our results are rigorous and expressed in the algebraic formulation of quantum statistical mechanics, see the Monographs \cite{bratteli_operator_1987,bratteli_operator_1997,naaijkens_quantum_2017}. The results take their roots in an ergodicity theorem proved in \cite{doyon2022hydro}, where quantum spin chains ($D=1$) with finite range interactions are discussed. A discussion of various physical aspects of the results and their links with other notions of ergodicity is given in \cite{ampelogiannis2021hydro}.

The main statement that we will prove is that the long-time averaging of any observable $A$ (any element of the $C^*$ algebra) converges, in the strong operator topology of the Gelfand-Naimark-Segal (GNS) representation of a state $\omega$, to its expectation in the state $\omega(A)$, times the identity operator. This holds for almost every speed with respect to the Lebesgue measure, and all rational directions on the lattice; we will refer to this as ``almost everywhere" or ``almost every space-time ray". The result holds in every space invariant (that is, invariant under spatial translations) factor state $\omega$ satisfying the Kubo-Martin-Schwinger (KMS) condition \cite[Section 5.3.1]{bratteli_operator_1997}. These include any KMS state of quantum spin lattice models with sufficiently fast-decaying interactions, at high enough temperature.

The result implies that convergence holds in every correlation function. Thus, connected correlation functions (Ursell functions) vanish in the limit, and the $n$-th moment of the ray-average of an observable converges to the $n$-th power of its expectation in the state. We will show, by using the Lieb-Robinson bound \cite{lieb_finite_1972}, that in fact, the results for two-point correlation functions and for the $n$-th moments hold in every space invariant factor state, not necessarily KMS.

The averaging can also be taken with respect to oscillatory factors, in a way that accounts for oscillations, for any wavelength and frequency. We show that this vanishes in the long time limit for almost every ray. We will also show that the property of the state being factor can be replaced by an appropriate property of clustering in space.

The result of time-averaging in KMS states is a strong, operatorial version of ergodicity. It indicates that the ray-average of an observable $A$ along almost every ray, including time-like rays of speed as small as desired, converges to its ensemble average $\omega(A)$. The result for the 2nd moment is a mean-square ergodicity result. It implies that the variance, with respect to the state, of a ray-averaged observable vanishes at long times, i.e.\ the ray-averaged observable does not fluctuate within the state (after a long time, almost everywhere).

Ergodicity -- the convergence under long-time averaging to ensemble averages -- is a difficult problem in many-body quantum systems. The present paper gives, we believe, the first general, rigorous results of this type that do not require any (generally hard-to-prove) assumptions of assymptotic abelianness in timelike directions. Further, in our understanding, the extension of the notion of ergodicity to averages with respect to oscillatory factors has not been considered previously in the literature. Our results show in particular that coherent asymptotic oscillatory behaviours of correlations are not possible on any set of rays covering a range of speeds of non-zero Lebesgue measure.

Ergodicity is closely related to the problem of {\em return to equilibrium}. This problem was first formulated, and basic results obtained, in \cite{robinson_return_1973,robinsonAlgebrasQuantumStatistical1976}. The connection of our results to some of these aspects is discussed in \cite{ampelogiannis2021hydro}.

The present ergodicity results indicate that a lot of information, initially encoded in the microscopic dynamics, is lost at larger scales. Since the results are general, for any $D$-dimensional quantum lattice model, they show that almost-everywhere ergodicity does not require particular properties of the dynamics, such as ``chaos". Rather, it is a consequence of ``extensivity", as encoded by the short-range interaction, and the fact that the state is factor (hence clusters at large distances).  

The loss of information at large scales of space-time can be formalised by the phenomenon of hydrodynamic projection: intuitively, the projection of observables onto conserved quantities at large scales. The phenomenon of hydrodynamic projection is proved rigorously in \cite{doyon2022hydro} for $D=1$ and finite range interactions, where it is seen to emerge as a consequence of the simplest statement of almost-everywhere ergodicity. One can extend these results to any $D$, as an application of some of the methods developed here. This involves quite a few additional mathematical constructions, and the statement and proof is reported in \cite{ampelogiannis2021hydro}; here, we concentrate on a full analysis of almost-everywhere ergodicity.

The paper is organised as follows. In Section \ref{sectsetup} we discuss the precise setup in which we work: the $C^*$ algebraic formulation for quantum spin lattices. In Section \ref{sectoverview} with give an overview of the main theorems, with the precise mathematical statements. In Section \ref{section:ergodicity_abelianness} we obtain simple but important results  outside the Lieb-Robinson cone, in space-like regions of the lattice. These will form the basis for the proof of ergodicity results within the Lieb-Robinson cone, which are provided in Sections \ref{sectproofergodicityprojection} - \ref{section:frequencyproof}. %Finally, Section \ref{sectapplication} discusses a potential connection of the ergodicity results to hydrodynamic projections.

\bmhead{Acknowledgments}
This work has benefited from discussions with Petros Meramveliotakis. BD is supported by EPSRC under the grant ``Emergence of hydrodynamics in many-body systems: new rigorous avenues from functional analysis", ref.~EP/W000458/1.  DA is supported by a studentship from EPSRC. 

\bmhead{Data availability statement}
Data sharing not applicable to this article as no datasets were generated or analysed during the current study.

\section{The set-up} \label{sectsetup}
We consider a lattice $\Z^D$, as a $\mathbb{Z}$-module, equipped with the $l_1$ norm. On each site $\boldsymbol n \in \Z^D$ it admits a quantum spin described by the (matrix) algebra of observables $\mathfrak{U}_{\boldsymbol n} \coloneqq M_{N_{\boldsymbol n}} ( \C)$,  $N_{\boldsymbol n} \in \N$ with the operator norm $||\cdot ||$, and for all lattice points $\boldsymbol n \in \Z^D$ the dimension of the spin matrix algebra is uniformly bounded, $N_{\boldsymbol n} \leq N$ for some $N \in \N$. To each finite $Λ \subset \Z^D$ we have the algebra 
\begin{equation}
    \mathfrak{U}_Λ\coloneqq \bigotimes_{\boldsymbol n \in Λ} \mathfrak{U}_{\boldsymbol n}
\end{equation}
and the algebra of local observables is defined  as the direct limit of the increasing net of algebras $\{ \mathfrak{U}_Λ \}_{Λ \in P_f(\Z^D)}$, where $P_f(\Z^D)$ denotes the set of finite subsets of $\Z^D$,
\begin{equation}
    \mathfrak{U}_{\rm loc} \coloneqq \lim_{\longrightarrow} \mathfrak{U}_Λ.
\end{equation}
The norm completion of $\mathfrak{U}_{\rm loc}$ defines the quasi-local C$^*$-algebra of the quantum spin lattice:
\begin{equation}
    \mathfrak{U} \coloneqq \overline{\mathfrak{U}_{\rm loc}}.
\end{equation}
For details on the definition of quasi-local C$^*$-algebras see \cite[Chapter 3.2.3]{naaijkens_quantum_2017}. 

Throughout this work we will use the following quantities:
The support of an observable $A \in \mathfrak{U}_{\rm loc}$, and the distance of the supports of $A,B \in \mathfrak{U}_{\rm loc}$ are defined as
\begin{equation}
    \supp(A) = \bigcap \{ X \subset \mathbb{Z}^D : A \in \mathfrak{U}_X \} 
\end{equation}
and
\begin{equation}
    \dist(A,B) \coloneqq \dist( \supp(A), \supp(B))
\end{equation}
respectively, where the distance between two subsets $X,Y$ of $\mathbb{Z}^D$ is 
\begin{equation}
    \dist(X,Y) = \min \{ |x-y| : x \in X , y \in Y \}.
\end{equation}
The diameter of subsets $X$ of $\mathbb{Z}^D$ is
\begin{equation}
    \diam(X) = \max \{ |x-y|: x,y \in X \}
\end{equation}
%Finally, the ball around $x_0$ of radius $d$ is \begin{equation}
%   B_{x_0}(d) = \{ x \in \mathbb{Z}^D : |x-x_0| \leq d \}
%\end{equation}
and, finally,  $|X|$ denotes the number of elements in $X$.

The dynamics are introduced by group actions on $\mathfrak{U}$, and  states of the system are the positive linear functionals:

\begin{defn}[Dynamical System] \label{defn:dynamicalsystem}
A dynamical system of a  quantum spin lattice is a triple $(\mathfrak{U},ι,τ)$ where $\mathfrak{U}$ is a quasi-local $C^*$-algebra, $τ$ is a strongly continuous representation of the group $\R$ by $^*$-automorphisms $\{ τ_t : \mathfrak{U} \xrightarrow{\sim} \mathfrak{U} \}_{t \in \R}$, and $ι$ is a representation of the translation group $\Z^D$ by $^*$-automorphisms $\{ ι_{\boldsymbol n}: \mathfrak{U} \xrightarrow{\sim} \mathfrak{U} \}_{\boldsymbol n \in \Z^D}$, such that for any $Λ \in P_f(\Z^D)$: $A \in \mathfrak{U}_Λ \implies ι_{\boldsymbol n }(A) \in \mathfrak{U}_{Λ+\boldsymbol n }$ for all $\boldsymbol n \in \Z^D$. We further assume that $τ$ is such that $τ_t ι_{\boldsymbol n} = ι_{\boldsymbol n} τ_t  $, $\forall t \in \R$, $\boldsymbol n \in \Z^D$, i.e.\ time evolution is homogeneous. 
\end{defn}

\begin{defn}[States, Invariance] \label{defn:state-invariance}
A state  of a dynamical system $(\mathfrak{U},ι,τ)$ is a positive linear functional $ω: \mathfrak{U} \to \C$ such that $\norm{ω}=1$. The set of states is denoted by $E_{\mathfrak{U}}$. A state $ω \in E_{\mathfrak{U}}$ is called space invariant if $ω( ι_{\boldsymbol n}(A) ) = ω(A)$, $\forall A \in \mathfrak{U}$, $\boldsymbol n \in \Z^D$ and time invariant if  $ω( τ_t(A) ) = ω(A)$, $\forall A \in \mathfrak{U}$, $t \in \R$. We will refer to a space and time invariant state simply as invariant.
\end{defn}
Space translations $ι_{\boldsymbol n}$ are naturally defined on $\mathfrak{U}_{\rm loc}$ and extended (by continuity) to $\mathfrak{U}$, see \cite[Chapter 3.2.6]{naaijkens_quantum_2017}. Time evolution is defined explicitly from the interaction of the quantum spin lattice, see \cite[Chapter 6.2]{bratteli_operator_1997}. An interaction is defined as a map $Φ: P_f(\Z^D) \to \mathfrak{U}_{\rm loc}$, s.t. $Φ(Λ) \in \mathfrak{U}_Λ$ and $Φ(Λ)=Φ^*(Λ)$, $\forall Λ \in P_f(\Z^D)$.  The Hamiltonian associated with any $Λ \in P_f(\Z^D)$ is
\begin{equation}
    H_Λ \coloneqq \sum_{X \subset Λ} Φ(X)
\end{equation}
which defines the local time evolution as
\begin{equation}
    τ_t^{Λ} (A) \coloneqq e^{itH_Λ}A e^{-itH_Λ},\quad A \in \mathfrak{U}_Λ ,\; t\in \R.
\end{equation}
Time evolution of the infinite system is defined when the limit $\lim_{Λ \to \infty}τ_t^{Λ}(A)$ exists in the norm for all $A \in \mathfrak{U}_{\rm loc}$ and can be uniquely extended to a strongly continuous $^*$-automorphism $τ_t$ of $\mathfrak{U}$. This can be proved for a large class of interactions, including finite range and exponentially decaying ones \cite[Theorem 6.2.11]{bratteli_operator_1997}. In particular, this holds for dynamical systems $(\mathfrak{U},ι,τ)$ with interaction that satisfies for some $λ>0$:
\begin{equation} 
    \norm{Φ}_λ \coloneqq \sup_{\boldsymbol n \in \Z^D} \sum_{X \ni \boldsymbol n} \norm{Φ(X)} |X| N^{2|X|} e^{λ \diam(X)} < \infty .
\end{equation}
which in the case of space-translation invariant interaction simplifies to 
\begin{equation} \label{eq:interaction}
    \norm{Φ}_λ \coloneqq  \sum_{X \ni \boldsymbol 0} \norm{Φ(X)} |X| N^{2|X|} e^{λ \diam(X)} < \infty .
\end{equation}
In what follows we restrict ourselves to such systems. \textbf{Our results hold for all quantum spin lattice models with interactions that satisfy \Cref{eq:interaction}.} This allows $Φ$ to include $m$-body interactions for any $m$, as long as the interaction drops at least exponentially with $m$ and with the distance. \textbf{This includes any finite range  and any two-body exponentially decaying interactions}.

An important tool, which we use extensively throughout this work, for dealing with C$^*$ algebras is the Gelfand–Naimark–Segal (GNS) representation, see \cite[Chapter 2.3]{bratteli_operator_1987} and in particular \cite[Theorem 2.3.16, Corollary 2.3.17]{bratteli_operator_1987}:
\begin{prop}[GNS representation] \label{prop:gns}
Given a state $ω \in E_{\mathfrak{U}}$ of a unital C$^*$ algebra $\mathfrak{U}$ there exists  a (unique, up to unitary equivalence)  triple $(H_ω,π_ω,Ω_ω)$ where $H_ω$ is a Hilbert space with inner product $\langle \cdot, \cdot \rangle$, $π_ω$ is a representation of the C$^*$ algebra by bounded operators acting on $H_ω$ and $Ω_ω$ is a cyclic vector for $π_ω$, i.e.\ the $\operatorname{span} \{ π_ω(A)Ω_ω : A \in \mathfrak{U}\}$ is dense in $H_ω$, such that
\begin{equation}
    ω(A) = \langle Ω_ω , π_ω(A) Ω_ω \rangle ,  \ \ A \in \mathfrak{U}.
\end{equation}
If additionally we have a group $G$ of automorphisms $\{ τ_g \}_{g \in G}$ of $\mathfrak{U}$ and $ω$ is $τ$-invariant, then there exists a representation of $G$ by unitary operators $U_ω(g)$ acting on $H_ω$. This representation is uniquely determined by 
\begin{equation}
    U_ω(g) π_ω(A) U_ω(g)^* = π_ω (τ_g (A)) \ , \ \ A \in \mathfrak{U}, g \in G
\end{equation}
and invariance of the cyclic vector
\begin{equation}
    U_ω(g)Ω_ω = Ω_ω \ , \forall g \in G.
\end{equation}
\end{prop}
The state of thermal equilibrium is characterised by the Kubo-Martin-Schwinger (KMS) condition, see \cite[Section 5.3.1]{bratteli_operator_1997}. We give the definition:
\begin{defn}[KMS state] \label{defn:KMS}
Consider a dynamical system $(\mathfrak{U},ι,τ)$. A state $ω$ is called a $(τ,β)$-KMS state, at inverse temperature $β$, if 
\begin{equation}
    ω(A τ_{iβ}B) = ω(BA) \label{eq:KMS_state}
\end{equation}
for all $A,B$ in a dense (with respect to the norm topology) $τ$-invariant $^*$-subalgebra of $\mathfrak{U}_τ$, where $\mathfrak{U}_τ$ is the set of entire analytic elements for $τ$.
\end{defn}
It is established that the set of analytic elements $\mathfrak{U}_τ$ is (norm) dense in $\mathfrak{U}$ for any strongly continuous one parameter group of automorphisms $τ$, see \cite[Proposition IV.4.6]{simon_statistical_2014}.  Hence, any dense subset of $\mathfrak{U}_τ$ is again dense in $\mathfrak{U}$.

\section{Overview of the main results}\label{sectoverview}

\subsection{Ergodic Theorems for space-time ray averaging} \label{subsection:ergodicity}

A central result in quantum spin lattice models, with interactions that decay sufficiently fast such as those we consider (\Cref{eq:interaction}), is the Lieb-Robinson bound, which limits the spread of the support of time-evolved observables, so that the ``strength" of the operator is exponentially decaying outside a light-cone defined by the Lieb-Robinson velocity. The bound was first stated in \cite{lieb_finite_1972} for finite range interactions. The Lieb-Robinson bound implies that certain states will be ergodic, in an appropirate sense, within a space-like region. We will call such states space-like ergodic:
\begin{defn}[Space-like ergodic state]
Consider a dynamical system $(\mathfrak{U}, ι, τ)$  as defined in \Cref{defn:dynamicalsystem}. A state $ω \in E_{\mathfrak{U}}$ is called space-like ergodic if it is space and time translation invariant and there exists a $υ_c>0$ such that for any $A,B \in \mathfrak{U}_{\rm loc}$, $\boldsymbol n \in \mathbb{Z}^D$ and $υ \in \hat{ {\mathbb{R}}}= \mathbb{R} \cup \{ -\infty, \infty \}$ with $|υ| > υ_c$ it holds that:
\begin{equation}
  \lim_N  \frac{1}{N} \sum_{m=0}^{N-1} ω \big( ι_{\boldsymbol n}^m τ_{υ^{-1}|\boldsymbol n|}^m (A) B \big) = ω(A) ω(B).
\end{equation}
\label{def:spacelike-ergodic}
\end{defn}

In \Cref{section:ergodicity_abelianness} we show that any factor (also called factorial, see \cite[p. 81, Def. 2.4.8]{bratteli_operator_1987}) invariant state is space-like ergodic with $υ_c=υ_{LR}>0$ the Lieb-Robinson velocity, in dynamical systems with interactions satisfying \eqref{eq:interaction}. This follows from the fact that the algebra $\mathfrak{U}$ is asymptotically abelian (in the sense of \cite[Def. 3]{kastler_invariant_1966}) for space-time translations inside a space-like region , i.e.\ $\norm{ [ι_{\boldsymbol{υ}t} τ_{t} A,B]} \rightarrow 0$ when $t\to \infty$ and $|\boldsymbol{υ}|>υ_{LR}$, by virtue of the Lieb-Robinson bound. For asymptotically abelian algebras it is established that factor states have clustering properties, see \cite[Example 4.3.24]{bratteli_operator_1987}. Hence, factor invariant states are an example of space-like ergodic states. Our goal is to show that the ergodic properties of such states also extend (almost everywhere) inside the light-cone $|\boldsymbol{υ}| < υ_{LR}$.

As any extremal (pure) state of a $C^*$ algebra is also factorial, the results hold for pure invariant states. Physically, this means that the results will hold for systems that are in a single thermodynamic phase. A special case is that of thermal states: at high enough temperatures there is a unique KMS state \cite[Proposition 6.2.45]{bratteli_operator_1997}.  This state is factorial, as an immediate consequence of \cite[ Theorem 5.3.30]{bratteli_operator_1997}.  By the result of existence of space invariant KMS states \cite[pg. 296]{bratteli_operator_1997}, this unique state will also be space invariant.  Note that time invariance is an immediate consequence of the KMS condition. KMS states are not limited to thermal states for the time evolution of the system. Indeed, one may consider time-stationary KMS states that are KMS with respect to a different ``time evolution", generated by a different short-range Hamiltonian, that commutes with the physical time evolution. These are naturally identified with ``generalised Gibbs ensembles" \cite{rigol_relaxation_2007}, in non-equilibrium many-body physics of integrable systems. \textbf{Hence our results will hold for any quantum spin lattice model, with sufficiently fast decaying interactions as expressed in \Cref{eq:interaction}, in a state of thermal equilibrium or a generalised Gibbs ensemble, at high enough (generalised) temperature.} 

At lower temperatures, there may exist multiple  thermodynamic phases described by the set of $(τ,β)$-KMS states, but this set is convex and its extremal points are  factor states \cite[ Theorem 5.3.30]{bratteli_operator_1997}. Hence,\textbf{ at lower temperatures the results still hold in space-invariant extremal KMS states}.

\textbf{The question that is raised, given that space-like ergodicity is satisfied in physically relevant states, is what happens within the light-cone?} There exist results concerning ergodicity but for the group action of space translations \cite[Section IV]{israel_convexity_1979}, can these be shown for time evolution?

We divide our results into two categories: those that do not necessitate the KMS condition, and those that do. The former are therefore more general, but the latter are much stronger.

\subsubsection{Results not requiring the KMS condition}

Our main theorem without using the KMS condition is that ergodicity still holds almost everywhere within the light-cone:

\begin{theorem}[Almost Everywhere Ergodicity]
Consider a dynamical system $(\mathfrak{U}, ι, τ)$, with interaction that satisfies \Cref{eq:interaction}  and a space-like ergodic state $ω\in E_{\mathfrak{U}}$, in the sense of Definition \ref{def:spacelike-ergodic}. It follows that for all $A,B \in \mathfrak{U}$, any rational direction $\boldsymbol{q} \in \Srat^{D-1} \coloneqq   \{\boldsymbol x / |\boldsymbol x|:\boldsymbol x\in \mathbb{Z}^D \setminus \{\boldsymbol0\} \}$ and \textbf{almost every speed} $v \in \mathbb{R}$ (with respect to the Lebesgue measure):
\begin{equation}
    \lim_{T \to \infty} \frac{1}{T} \int_0^T  ω \bigg( ι_{\floor{ \boldsymbol{v}t}}τ_t (A) B \bigg)  \,dt = ω(A) ω(B) \label{eq:maintheorem}
\end{equation}
where $\boldsymbol{v}= v \boldsymbol{q}$, and where for any vector $\boldsymbol{a}=(a_1,a_2,...,a_D)$ we denote $\floor{\boldsymbol{a}} \coloneqq ( \floor{a_1},...,\floor{a_D})$.
\label{th:maintheorem}
\end{theorem}

The proof is done in \Cref{section:ergodicityproof}. In the end we show that \textbf{in the case of KMS states this result can be extended to any number of observables $B_j$ and ray averaged of observables $A_i$, see \Cref{th:general_theorem}}.

It is simple to see that \Cref{eq:maintheorem} will hold for any $v>v_c$, by the assumption of space-like ergodicity. There are general results for going from discrete to continuous time ergodic theorems, see \cite{bergelson_form_2011}. However, for $υ<υ_c$ this is non-trivial and we are not aware of any previous results that deal with this situation, besides \cite{doyon2022hydro} in the case of one-dimensional, finite-range interactions.

A few remarks about the physics of the above theorem, and of our other main results below, are in order. The equality of time-averaged quantities with their ensemble averages is a notion of ergodicity. However, conventionally ergodicity does not involve space translation: the ray is at velocity 0. The above theorem, and all our main results, are notions of ergodicity along almost every ray, but not necessarily at zero velocity (see Ref.~\cite{ampelogiannis2021hydro} where our present ergodic theorems are put within the context of conventional discussions on ergodicity and of hydrodynamics, and some of the techniques developed here are applied to conventional ergodic statements). What is the physical meaning or motivation for ergodicity at other velocities?

First, our ergodic theorems can be seen in a somewhat natural fashion as a ``thinning'' of the time evolved observable $τ_t (A)$, as it spatially spreads, over time, within the Lieb-Robinson light-cone. Indeed, looking within the light-cone at different velocities, $ ι_{\floor{ \boldsymbol{υ}t}}τ_t (A)$, we see that the time-average does not correlate with, hence is ``not seen" by, any other observable within the state $ω$, for almost every velocity; it may only be seen on a set of velocities of measure zero. This thinning property is made stronger by the operatorial version of the result below, \Cref{th:KMS_average_strong_convergence}. This sheds light on the problem of operator spreading discussed in the literature (see e.g.~\cite{nahum2018operator,keyser2018operator,khemani2018operator,rakov2018diffusive}): we see that under time averaging, much of the structure of operator growth is washed out.

Second, the notion that time averages reproduce ensemble averages is based on the physical idea that a measure of some observable, somewhere in space, in fact will naturally be time averaged, because of the finite measurement time. From this viewpoint, one may argue that time averaging does not have to be in pure-time direction: for instance, a gas may be in a Galilean-boosted thermal state (there may be a nonzero overall velocity), and time averaging is equivalent to ray averaging of the non-boosted gas. Then, the question of ergodicity at almost every velocity is the relevant one. Our results are for lattice gases; it will be interesting to study systems on continuous space, and almost-every-boost ergodicity results.

Third, one of the most important applications of Theorem \ref{th:maintheorem} is to hydrodynamic projections \cite{doyon2022hydro,ampelogiannis2021hydro}. Because hydrodynamics only makes predictions for fluid-cell averaged observables, it turns out that it is sufficient to have ergodicity at almost every velocity in order to derive general hydrodynamic projection results (see \cite{ampelogiannis2021hydro} for the details).

Finally, one may ask about the meaning of the measure-zero sets $\mathcal{K}_{\boldsymbol q}$ (one for every direction $\boldsymbol{q}\in\Srat^{D-1})$ of velocities $v$ where \eqref{eq:maintheorem} fails. For generic interacting models, we expect these sets to be empty -- all local observables decorrelate in time, either exponentially or, at the speeds of hydrodynamic modes, algebraically. But examples of models where a $\mathcal{K}_{\boldsymbol q}$ is not empty are the trivial lattice models without inter-site interaction (only site self-energy): one can easily construct such models where some single-site observable $A$ is time invariant. In fact, these are particular examples of {\em constrained models}, and we believe more general, interacting constrained models, actively studied recently (see e.g.~\cite{PhysRevX.10.021051}), may have non-trivial $\mathcal{K}_{\boldsymbol q}$. It would be interesting to investigate this further.

\begin{remark}
It is worth noting that the restriction to the rational directions, which has  measure $0$ in $\mathbb{R}^{D-1}$, is rather natural: as we are on a discrete space, the relevant directions are the ones in $\Srat^{D-1}$, not $\mathbb{S}^{D-1}$ (the unit sphere in $\mathbb{R}^D$). Hence, ``almost everywhere'' indeed means for almost every space-time ray of the lattice: space is discrete. Further, for applications to hydrodynamic projection theorems \cite{ampelogiannis2021hydro}, it turns out that considering rational directions is sufficient, as these are dense in $\mathbb S^{D-1}$. One may nevertheless ask what happens along irrational directions, which we leave for future work.
\end{remark}

Ergodic theorems usually involve the convergence of series (or integrals) of operators into projections. In our case, we will express  \Cref{eq:maintheorem} in the GNS representation (\Cref{prop:gns}) and use von Neumann's mean ergodic Theorem \ref{th:neumann} ( see also \cite[Theorem II.11]{reed_i_1981}) to get a projection.  This projection will be shown  to have rank one almost everywhere, by use of space-like erogidicity, the group properties of time and space translations and the separability of the GNS Hilbert space of a quantum spin lattice:

\begin{theorem}[Projection]
Consider a dynamical system $(\mathfrak{U}, ι, τ)$ with interaction that satisfies \Cref{eq:interaction} and a space-like ergodic state $ω\in E_{\mathfrak{U}}$, in the sense of Definition \ref{def:spacelike-ergodic}. Let $(H_ω,π_ω,Ω_ω)$ be the GNS representation associated to $ω$ and consider the unitary representation $U_ω(\boldsymbol n,t)$ of the group $G=\mathbb{Z}^D \times \mathbb{R}$  acting on $H_ω$ (\Cref{prop:gns}). It follows for any $\boldsymbol n \in \mathbb{Z}^D$, that for almost all $υ \in \mathbb{R}$ the orthogonal projection $P_{υ,\boldsymbol n}$ on the subspace of  $H_ω$ formed by vectors invariant under $U_ω(\boldsymbol n,υ^{-1}|\boldsymbol n|)$ is the rank one projection on $Ω_ω$ .
\label{th:rankone}
\end{theorem}

In fact there is an equivalence between Eq. \ref{eq:maintheorem} and $P_{υ,\boldsymbol n}$ being rank one, as will be made clear by the proof, see also \cite[Section 2]{fidaleo2014Nonconventional}. Relevant results include \cite[Theorems 4.3.17, 4.3.20, 4.3.22]{bratteli_operator_1997} and \cite[Proposition 6.3.5]{ruelle_statistical_1974}, which, however,  assume some form of asymptotic abelianness or clustering for the whole group, while we have it only in a space-like region.

The proofs of the Ergodicity and Projection Theorems are done in \Cref{section:ergodicityproof}. A related result, proven in \Cref{section:meansquaredproof},  is the vanishing of the variance, with respect to the state, of ray averaged observables at the long time limit $T \to \infty$:
\begin{theorem}[Mean-square ergodicity] \label{th:meansquared}
 Consider a dynamical system $(\mathfrak{U}, ι, τ)$  with interaction that satisfies \Cref{eq:interaction}   and a space-like ergodic state $ω\in E_{\mathfrak{U}}$, in the sense of Definition \ref{def:spacelike-ergodic}. It follows that for all $A,B \in \mathfrak{U}$, any rational direction $\boldsymbol{q} \in \Srat^{D-1}$ and almost every speed $v \in \mathbb{R}:$
\begin{equation}
    \lim_{T \to \infty} \frac{1}{T^2}\int_0^{T} \int_0^{T} \ ω\big( ι_{\floor{ \boldsymbol{v} t_1}}τ_{t_1} (A)  \   ι_{\floor{ \boldsymbol{v} t_2}}τ_{t_2} (B)   \big) \,dt_1 \, dt_2= ω(A) ω(B)
\end{equation}
%\begin{equation}
%    \lim_{T,T^{\prime} \to \infty} \frac{1}{TT^{\prime}}\int_0^{T^{\prime}} dt_1 \int_0^{T} dt_2 \ ω\big(ι_{\floor{ \boldsymbol{υ} t_1}}τ_{t_1} (A)  ι_{\floor{ \boldsymbol{υ} t_2}}τ_{t_2} (A)  \big)= \big(ω(A) \big)^2 
%\end{equation}
with $\boldsymbol{v}= v \boldsymbol{q}$. In fact, this holds as a double limit $\lim_{T,T^{\prime} \to \infty} \frac{1}{T^{\prime} T}\int_0^{T}\int_0^{T^{\prime}}  (\cdots) \, dt  \, dt^{\prime}$.
\end{theorem}
We note that the usual expression of mean-square ergodicity is for the choice $A=B$. Further, we can extend this to any $n$-th moments with respect to the state:
\begin{theorem} \label{th:meanN}
Under the assumptions of \Cref{th:meansquared} we have that for any $n \in \N$,
\begin{equation}
    \lim_{T \to \infty} \frac{1}{T^n} ω \bigg( \big(\int_0^T  ι_{\floor{ \boldsymbol{v} t}}τ_{t} A  \,dt  \big)^n \bigg)= \big( ω(A) \big)^n.
\end{equation}
Again, this also holds also as a multiple limit on $T_1,T_2,\ldots, T_n$, with different observables $A_1,A_2,\ldots,A_n\in\mathfrak U$:
\begin{equation}
    \lim_{T_1,\ldots,T_n\to\infty}
    \omega\big(\prod_{j=1}^n 
    \frac1{T_j} \int_0^{T_j} dt_j \,ι_{\floor{ \boldsymbol{v} t_j}}τ_{t_j} A_j\big)= 
    \prod_{j=1}^n \omega(A_j).
\end{equation}
The integral $\int_0^T  ι_{\floor{ \boldsymbol{v} t}}τ_{t} A  \,dt  $ is well defined as a Bochner integral, see \cite[Section 3.7]{hille_functional_1996} and \Cref{appendix1}.
\end{theorem}
To be precise, the multiple limit is in the sense that $(T_1^{-1},T_2^{-1},\ldots,T_{n+1}^{-1})\to \boldsymbol 0$ with respect to the topology of $\R^{n+1}$. See \Cref{appendix:Moore-Osgood} for details on how we treat multiple limits.

Here we have written the result by putting the integral within the state evaluation, using the Bochner integral discussed in \Cref{appendix1}. This makes the meaning of the result more clear: the ray-averaged observable $A$ does not fluctuate within the state $ω$ in the long time limit.

\subsubsection{Results for KMS states}

In the case of space-like ergodic KMS states, we can take advantage of the KMS property, \Cref{eq:KMS_state}, to prove much stronger results. In particular, we can show that \textbf{in the GNS representation of a KMS state the ray average of an observable tends (in the strong operator topology) to the state average times the identity operator, in the long time limit, almost everywhere.} The result will hold \textbf{for any space-like ergodic state that is also a $(τ^{\prime},β)$-KMS state for any time evolution $τ^{\prime}$ which commutes with the initial time evolution of the system $τ$}.

Few remarks are in order concerning the importance of admitting $\tau'\neq \tau$.

First, in the context of modern non-equilibrium physics, especially in integrable systems, KMS states with respect to different time evolutions are of fundamental importance. Indeed, as we mentioned, these should be identified with generalised Gibbs ensembles (GGE) \cite{rigol_relaxation_2007}, which formally have density matrices $e^{-\sum_a \beta_a Q_a}$ where $Q_a$ are some set of short-range conserved charges for the evolution Hamiltonian $H$; here $\sum_a \beta_a Q_a = H'$ would be the short-range Hamiltonian generating $\tau'$. GGEs are the states which a system is expected to reach after quantum quenches and in other non-equilibrium situations at large times. The currently most rigorous understanding of GGEs in interacting systems \cite{doyonPseudolocal} has not yet been connected with the KMS condition, and it would be interesting to investigate this in future works.

Second, mathematically, using the Tomita-Takesaki theory \cite{summers2005tomitatakesaki}, every faithful normal state of a von Neumann Algebra \cite[Section 2.2]{bratteli_operator_1987} is a KMS state with time evolution given by the modular group. Adapting our proofs and results to the context of von Neumann algebras, we would thus have general results for space-like ergodic, faithful, normal states, with $\tau'$ the modular group. Of course, closing the $C^*$ algebra of operators on $H_\omega$ under the strong topology gives a natural von Neumann algebra to work with. This requires some additional technical work (e.g. extending our results to $W^*$ dynamical systems), which we leave for future works.

Finally, as an aside, making the connection with GGEs, we note that one may wish to identify the set of all space-time stationary, faithful, normal states with (at least some large subset of) GGEs, the modular group being the group generated by the particular charge $\sum_a \beta_a Q_a = H'$ for the GGE. Thus the Tomita-Takesaki theory could give a very general and mathematically useful understanding of GGEs.

Our main theorem using the KMS condition is the following.
\begin{theorem} \label{th:KMS_average_strong_convergence}
Consider a dynamical system $(\mathfrak{U}, ι, τ)$ and a possibly different time evolution $\tau'$ such that $(\mathfrak{U}, ι, τ')$ is also a dynamical system. Suppose the interactions, both for $\tau$ that $\tau'$, satisfy \Cref{eq:interaction}, and that $\tau$ and $\tau'$ commute. Consider a   $(\tau',\beta)$-KMS state $ω_β$ that is space-like ergodic for  $(\mathfrak{U}, ι, τ)$  and its respective GNS representation $(H_{ω_β},π_{ω_β},Ω_{ω_β})$. It follows that for any $A \in \mathfrak{U}$, its ray average in the GNS representation converges to the state average in the strong operator topology, for almost every ray. That is, the following limit holds in the norm (of $H_{ω_β}$)
\begin{equation}
    \lim_{T \to \infty} \frac{1}{T} \int_0^T π_{ω_β} \big( ι_{\floor{ \boldsymbol{υ} t}}τ_t A \big) \, dt \, Ψ= ω_β(A) Ψ \ , \ \ \forall Ψ \in H_{ω_β}
\end{equation}
with $\boldsymbol{υ}= υ \boldsymbol q$ of any rational direction $\boldsymbol{q} \in \Srat^{D-1}$  and almost every speed $υ\in \R$. Again, the integral is to be understrood as a Bochner integral (\Cref{appendix1}). 
\end{theorem}
\textbf{This can be understood as convergence, in the strong operator topology of the GNS Hilbert space, of the ray average to the ensemble average, along almost every space-time ray (of rational direction)}:
\begin{equation}
    \frac{1}{T} \int_0^T π_{ω_β} \big( ι_{\floor{ \boldsymbol{υ} t}}τ_t A \big) \, dt  \xrightarrow{SOT} ω_β(A) \mathds{1} \ , \text{ almost everywhere.}
\end{equation}

\begin{remark}
Note that in the case $τ^{\prime}=τ$ and of small $β$ (high temperatures) it suffices to assume that the state is only KMS, as this immediately implies space-like ergodicity, see the discussion above \Cref{th:maintheorem}. For higher values of $β$ where the KMS state is not unique, any KMS state that is invariant and factor will subsequently be space-like ergodic.
\end{remark}
The proof of this Theorem is done in \Cref{section:strong_convergence}. In fact, as is made clear by the proof, we need less than the space-like ergodic state being KMS for the Theorem to hold. It would suffice to have a space-like ergodic state that satisfies a KMS-like condition, i.e.\ that $ω(A τ^{\prime} B) = ω(B A)$ for any linear map $τ^{\prime}: \mathfrak{U_{τ^{\prime}}} \to \mathfrak{U}$ that commutes with the time evolution, where $ \mathfrak{U_{τ^{\prime}}}$ is a dense subset of $\mathfrak{U}$ that is $\tau$-invariant, $\tau_t(\mathfrak{U_{τ^{\prime}}}) = \mathfrak{U_{τ^{\prime}}}\;\forall t$.

In the study of operator algebras, and in many physical applications, it is often useful to have results for commutators of observables (for instance, this is how the Lieb-Robinson bound was originally stated). Let us therefore discuss some results regarding the long time ray-averages of commutators of observables. An immediate consequence of the Ergodicity Theorem \ref{th:maintheorem} is the vanishing of the ray average of the commutator within the state:
\begin{corollary}
Consider the assumptions of \Cref{th:maintheorem}. It follows that for all $A,B \in \mathfrak{U}$, any rational direction $\boldsymbol{q} \in \Srat^{D-1}$ and almost every speed $υ \in \mathbb{R}:$
\begin{equation}
       \lim_{T \to \infty} \frac{1}{T} \int_0^T  ω \big( [ι_{\floor{ \boldsymbol{υ} t}}τ_t (A) ,B ]\big)  \,dt = 0  \end{equation}
      with $\boldsymbol{υ}= υ \boldsymbol{q}$ and $[ \cdot, \cdot ]$ denoting the commutator.
\end{corollary}
However, we can show a  stronger result regarding the commutator in the case of space-like ergodic KMS states. Again, under the same set-up as \Cref{th:KMS_average_strong_convergence}, we can consider space-like ergodic states that satisfy the KMS condition with respect to any time evolution $τ^{\prime}$ that commutes with $τ$ and satisfies \Cref{eq:interaction}. It follows that we have mean asymptotic abelianness, for space-time translations along almost every ray, in the  GNS representation:

\begin{theorem}[Mean asymptotic abelianness in the GNS representation of a KMS state] \label{th:mean_asymptotic_gns_abelianness}
Consider the assumptions of Theorem \ref{th:KMS_average_strong_convergence}. We have that in the GNS representation of $ω_β$, for all $B,C \in \mathfrak{U}$, any rational direction $\boldsymbol{q} \in \Srat^{D-1}$   of the velocity $\boldsymbol{υ}= υ\boldsymbol{q}$ and almost every speed  $υ \in \mathbb{R}$:
\begin{equation}
   \lim_{T \to \infty} \frac{1}{T} \int_0^T π_{ω_β} ( [ι_{\floor{ \boldsymbol{υ} t}}τ_t (B) ,C ] ) \,dt \,Ψ =0 \ , \ \  \forall Ψ \in H_{ω_β}
\end{equation}
where the integral is to be understood as a Bochner integral and the limit is in the Hilbert space norm.

\end{theorem}
\textbf{This can be understood as convergence, in the strong operator topology of the GNS Hilbert space, of the ray-averaged commutator along almost every space-time ray (of rational direction):}
\begin{equation}
    \frac{1}{T} \int_0^T π_{ω_β} ( [ι_{\floor{ \boldsymbol{υ} t}}τ_t (B) ,C ] ) \,dt  \xrightarrow{SOT} 0 \ \text{, almost everywhere.}
\end{equation}

This result is an immediate consequence of \Cref{th:KMS_average_strong_convergence} and is shown in \Cref{section:mean_asympt_abel}.

Finally, in KMS states the ergodicity \Cref{th:maintheorem} can be extended to include an arbitrary number of observables and of ray-averaged observables, in any order:
\begin{theorem}  \label{th:general_theorem}
Consider the assumptions of Theorem \ref{th:KMS_average_strong_convergence} and $n \in \N$. Consider also $A_1, A_2 , \cdots , A_{n+1} \in \mathfrak{U}$ and $B_1, B_2, \cdots, B_n \in \mathfrak{U}$ and denote $\overline{B_j^T} = \frac{1}{T} \int_0^T ι_{\floor{ \boldsymbol{υ} t}}τ_t (B_j) \,dt$ their ray-average, $T \in (0, \infty)$. It follows that for any rational direction $\boldsymbol{q} \in \Srat^{D-1}$ of the velocity $\boldsymbol{v}=υ \boldsymbol{q}$ and almost every speed  $υ \in \mathbb{R}$, the following multiple limit exists and gives:
\begin{equation} \label{eq:general_theorem}
\arraycolsep=1.4pt\def\arraystretch{2.2}
\begin{array}{*3{>{\displaystyle}lcl}p{5cm}}
    \lim_{T_1,\ldots,T_n \to \infty}ω \big(  A_1   \overline{B_1^{T_1}} A_2 \overline{B_2^{T_2}}  \cdots \overline{B_n^{T_n}}  A_{n+1} \big) = \\ =ω(B_1) ω(B_2) \cdots ω(B_n) ω( A_1 A_2 \cdots A_{n+1}).
    \end{array}
\end{equation}
\end{theorem}
This is shown in \Cref{section:proof_of_general_theorem} by induction. The cases $n=1$ and $n=2$ are shown in \Cref{lemma:6.1} and subsequently used in proving the SOT convergence of \Cref{th:KMS_average_strong_convergence}.
\begin{remark}
Note that \Cref{eq:general_theorem} means that we can have any number of observables $A_i$ and ray-averaged ones $\overline{B_j^T}$ in any possible combination, i.e.\ by choosing for any  $A_i= A_{i_1}A_{i_2} \cdots A_{i_k}$ for arbitrary $A_{i_1},A_{i_2} \cdots, A_{i_k} \in \mathfrak{U}$, $k \in \N$ we see that we can have any number of observables between the ray-averaged ones. Note also that any of the $A_i$ can be chosen to be the unit element of $\mathfrak{U}$.
\end{remark}

\subsection{Ergodic Theorems taking account of oscillations} \label{subsection:frequency}

The ray average can also be taken in such a way so as to account for oscillations, in order to show that correlations between observables cannot sustain certain types of oscillations.
The main assumption is what we call space-like non-oscillating states, which will again follow from the Lieb-Robinson bound, in factor states. The intuition is that we assume states whose correlation functions do not exhibit oscillatory behaviours inside a space-like cone:

\begin{defn}[Space-like non-oscillating state]
Consider a dynamical system $(\mathfrak{U}, ι, τ)$. A state $ω \in E_{\mathfrak{U}}$ is called space-like non-oscillating if it is space and time translation invariant and there exists a $υ_c>0$ such that for any $A,B \in \mathfrak{U}_{\rm loc}$, $\boldsymbol n \in \mathbb{Z}^D$ and $υ \in \hat{ {\mathbb{R}}}= \mathbb{R} \cup \{ -\infty, \infty \}$ with $|υ| > υ_c$ it holds that:
\begin{equation}
  \lim_N  \frac{1}{N} \sum_{m=0}^{N-1} e^{-i(\boldsymbol k \cdot \boldsymbol n - f υ^{-1} |\boldsymbol n|)m} ω \big( ι_{\boldsymbol n}^m τ_{υ^{-1}|\boldsymbol n|}^m (A) B \big) =0
\end{equation}
for all wavenumber-frequency pairs $(\boldsymbol k ,f )\in \R^D \times \R$ with $\boldsymbol k \cdot \boldsymbol n - f υ^{-1} |\boldsymbol n| \not\in 2\pi \Z$.
\label{def:phase-spacelike-ergodic}
\end{defn}

We will see that factor invariant states indeed satisfy this in the systems that we consider, see \Cref{th:space-like-clustering} and below.

\subsubsection{Results not requiring the KMS condition}

Having the assumption of a space-like non-oscillating state we can extend the Ergodicity  \Cref{th:maintheorem} to preclude oscillatory behaviours of correlations along almost every ray:
\begin{theorem}
Consider a dynamical system $(\mathfrak{U}, ι, τ)$ with interaction that satisfies \Cref{eq:interaction}, a space-like non-oscillating state  $ω\in E_{\mathfrak{U}}$ and any $A,B \in \mathfrak{U}$.  Then for every non-zero $(\boldsymbol k, f)\in \mathbb{R}^D \times \mathbb{R}$ and every rational  direction $\boldsymbol q \in \Srat^{D-1}$ of the velocity $\boldsymbol{υ}= υ \boldsymbol{q}$ it follows that  for almost all $υ\in \R$
\begin{equation}
   \lim_{T \to \infty} \frac{1}{T} \int_0^T  e^{i (\boldsymbol k\cdot \boldsymbol v   -  f) t} \bigg( ω \big( ι_{\floor{ \boldsymbol{υ}t}}τ_t (A) B \big) - ω(A) ω(B) \bigg)  \,dt = 0.
    \label{eq:frequencyav2}
\end{equation}
Note that generally this means $\lim_{T \to \infty} \frac{1}{T} \int_0^T  e^{i (\boldsymbol k\cdot \boldsymbol v   -  f) t}  ω \big( ι_{\floor{ \boldsymbol{υ}t}}τ_t (A) B \big) \,dt =0$, except for the special case $\boldsymbol k\cdot \boldsymbol v-f=0$ where the Theorem reduces to $\Cref{th:maintheorem}$.

Additionally, we have for any $n$-th moment with respect to the state:
\begin{equation} \label{eq:frequency-moments}
    \lim_{T \to \infty} \frac{1}{T^n} ω \bigg( \big(\int_0^T e^{i (\boldsymbol k\cdot \boldsymbol v   -  f) t}   ι_{\floor{ \boldsymbol{υ} t}}τ_{t} A  \,dt  \big)^n \bigg)=  
   \begin{cases}
		ω(A)^n,  & \mbox{if } \boldsymbol k \cdot \boldsymbol v - f =0 \\
		0, &  \text{otherwise}
	\end{cases}
\end{equation}
for almost all $υ \in \mathbb{R}$. \label{th:frequencyav2}. 
\end{theorem}

This  tells us that both correlations and moments cannot sustain oscillations in the long time limit, along almost every ray. The essential part of the proof, which is done in \Cref{section:frequencyproof}, will again be a projection:
\begin{theorem} \label{th:frequencyproj}
Consider a dynamical system $(\mathfrak{U}, ι, τ)$ with interaction that satisfies \Cref{eq:interaction}, a space-like non-oscillating state $ω\in E_{\mathfrak{U}}$, the respective GNS representation $(H_ω, π_ω, Ω_ω)$ and the unitary representation $U_ω(\boldsymbol n,t)$. Let $υ\in \R$, $\boldsymbol n \in \Z^D$ and $(\boldsymbol k, f) \in \R^{D} \times \R$. We denote by $P_{υ,\boldsymbol n}(\boldsymbol k, f)$ the projection on the subspace of $H_ω$ formed by vectors invariant under the operator $ e^{-i(\boldsymbol k \cdot \boldsymbol n - f υ^{-1} |\boldsymbol n|)}U_ω( \boldsymbol n, υ^{-1}|\boldsymbol n|)$. Then for every $(\boldsymbol k, f) \in \R^D \times \R$ and every $\boldsymbol n \in \mathbb Z^D$  it follows that for almost every $υ\in\mathbb R$, 
$P_{υ,\boldsymbol n}(\boldsymbol k, f)=0$ whenever $\boldsymbol k\cdot \boldsymbol n - fυ^{-1} |\boldsymbol n|\not\in 2\pi \mathbb{Z}$.
%Then it follows that $P_{υ,\boldsymbol n}(\boldsymbol k, f)=0$ for all $\boldsymbol n \in \Z^D$, $(\boldsymbol k, f) \in \R^D \times \R$ and for almost all $υ \in \R$ such that $\boldsymbol k\cdot \boldsymbol n - fυ^{-1} |\boldsymbol n|\not\in2\pi\mathbb{Z}$.

\end{theorem}
\iffalse%%%%%%%%%%%%%%%%%%%%%%%%%%%%%%%%%%%%%%%%%%%%%%%%%%%%%%%%%%%%%%%%%%%%%%%%%%%%%
\begin{theorem} \label{th:frequencyproj}
Consider a dynamical system $(\mathfrak{U}, ι, τ)$ with interaction that satisfies \Cref{eq:interaction}, a space-like non-oscillating state $ω\in E_{\mathfrak{U}}$, the respective GNS representation $(H_ω, π_ω, Ω_ω)$ and the unitary representation $U_ω(\boldsymbol n,t)$. Let $υ\in \R$, $\boldsymbol n \in \Z^D$ and $(\boldsymbol k, f) \in \R^{D-1} \times \R$. We denote by $P_{υ,\boldsymbol n}(\boldsymbol k, f)$ the projection on the subspace of $H_ω$ formed by vectors invariant under the operator $ e^{-i(\boldsymbol k \cdot \boldsymbol n - f υ^{-1} |\boldsymbol n|)}U_ω( \boldsymbol n, υ^{-1}|\boldsymbol n|)$. Then:
\begin{enumerate}
    \item For any $υ \in \R$, $\boldsymbol n \in \Z^D$, it follows that $P_{υ,\boldsymbol n}(\boldsymbol k, f)=0$ for almost all $(\boldsymbol k,f)$. %\label{th:frequencyproj1}
    \item If $ω$ is also space-like non-oscillating  then for any $(\boldsymbol k,f)$, it follows that $P_{υ,\boldsymbol n}(\boldsymbol k, f)=0$ for almost all $υ \in \R$ and all $\boldsymbol n \in Z^D$.
\end{enumerate}
\end{theorem}
\fi%%%%%%%%%%%%%%%%%%%%%%%%%%%%%%%%%%%%%%%%%%%%%%%%%%%%%%%%%%%%%%%%%%%%%%%%%%%%%%%%%%%%%%

\subsubsection{Results for KMS states}

The SOT convergence results \Cref{th:KMS_average_strong_convergence}, \Cref{th:mean_asymptotic_gns_abelianness} in the GNS representation of KMS states can also be extended to account for oscillations. The proofs are very similar and hence are omitted.

\begin{theorem} \label{th:kms_strong_convergence_oscil}
    Consider a dynamical system $(\mathfrak{U}, ι, τ)$ and a possibly different time evolution $\tau'$ such that $(\mathfrak{U}, ι, τ')$ is also a dynamical system. Suppose the interactions, both for $\tau$ that $\tau'$, satisfy \Cref{eq:interaction}, and that $\tau$ and $\tau'$ commute. Consider a   $(\tau',\beta)$-KMS state $ω_β$ that is space-like non-oscillating for  $(\mathfrak{U}, ι, τ)$  and its respective GNS representation $(H_{ω_β},π_{ω_β},Ω_{ω_β})$. It follows that for any $A\in \mathfrak{U}$, any $(\boldsymbol k,f)\in \R^D \times \R$ and any rational direction $\boldsymbol q\in \Srat^{D-1}$ of the velocity $\boldsymbol{υ}=υ\boldsymbol{q}$
\begin{equation}
    \lim_{T \to \infty} \frac{1}{T} \int_0^T e^{i (\boldsymbol k\cdot \boldsymbol v   -  f) t}\bigg(π_{ω_β} \big( ι_{\floor{ \boldsymbol{υ} t}}τ_t A \big) -ω(A) \bigg) \, dt \, Ψ= 0\ , \ \ \forall Ψ \in H_{ω_β}
\end{equation}
for almost every speed $υ \in \R$.  The limit is in the norm of the GNS Hilbert space.  Note that for generic $(\boldsymbol k,f)$ this means $\lim_{T \to \infty} \frac{1}{T} \int_0^T e^{i (\boldsymbol k\cdot \boldsymbol v   -  f) t}π_{ω_β} \big( ι_{\floor{ \boldsymbol{υ} t}}τ_t A \big)  \, dt \, Ψ= 0\ $, precluding sustained oscillations of the ray averaged operator.
\end{theorem}

\begin{theorem}
Consider the assumptions of Theorem \ref{th:kms_strong_convergence_oscil}.  It follows that for all $B,C\in \mathfrak{U}$, for any rational direction $\boldsymbol{q} \in \Srat^{D-1}$ of the velocity $\boldsymbol{v}=υ \boldsymbol{q}$ and almost every speed  $υ \in \mathbb{R}$:
\begin{equation}
       \lim_{T \to \infty}  \frac{1}{T} \int_0^T e^{i (\boldsymbol k\cdot \boldsymbol v   -  f) t}π_{ω_β} ( [ι_{\floor{ \boldsymbol{υ} t}}τ_t (B) ,C ] ) \,dt \,Ψ =0 \ , \ \  \forall Ψ \in H_{ω_β}
\end{equation}
where the limit is in the norm of the GNS Hilbert space. This precludes sustain oscillations of the ray averaged commutator.
\end{theorem}

We can also extend the n-point \Cref{th:general_theorem} in a similar way, to show:
\begin{theorem}  \label{th:general_theorem_oscil}
Consider the assumptions of Theorem \ref{th:kms_strong_convergence_oscil} and $n \in \N$. Consider also $A_1, A_2 , \cdots , A_{n+1} \in \mathfrak{U}$ and $B_1, B_2, \cdots, B_n \in \mathfrak{U}$ and denote $\widetilde{B_j^T} = \frac{1}{T} \int_0^T e^{i (\boldsymbol k\cdot \boldsymbol v   -  f) t}ι_{\floor{ \boldsymbol{υ} t}}τ_t (B_j) \,dt$ their ray-average, $T \in (0, \infty)$, $(\boldsymbol k,f)\in \R^D \times \R$ . It follows that for any $(\boldsymbol k,f)\in \R^D \times \R$, any rational direction $\boldsymbol{q} \in \Srat^{D-1}$ of the velocity $\boldsymbol{v}=υ \boldsymbol{q}$ and almost every speed  $υ \in \mathbb{R}$:
\begin{equation} \label{eq:general_theorem_oscil}
 \begin{array}{*3{>{\displaystyle}lcl}p{5cm}}
    \lim_{T \to \infty}  &ω&\big(  A_1   \widetilde{B_1^T} A_2 \widetilde{B_2^T}  \cdots \widetilde{B_n^T}  A_{n+1} \big)  \\
   &=&\begin{cases}
		ω(B_1) ω(B_2) \cdots ω(B_n) ω( A_1 A_2 \cdots A_{n+1}) ,  &  \boldsymbol k \cdot \boldsymbol v - f =0 \\
		0, &  \text{otherwise.}
	\end{cases} 
	\end{array}
\end{equation}
\end{theorem}

Finally,  under only the assumption of state invariance we can show \Cref{eq:frequencyav2} for all $υ$ and almost all $(\boldsymbol k,f)$, using the simple argument of countability of eigenvalues:
\begin{theorem} \label{th:frequencyaverage1}
Consider a dynamical system $(\mathfrak{U}, ι, τ)$ with interaction that satisfies \Cref{eq:interaction}, a $ι, τ$-invariant state $ω\in E_{\mathfrak{U}}$ and any $A,B \in \mathfrak{U}$. Then, for every velocity $\boldsymbol{υ}= υ \boldsymbol q$ of rational direction, $υ\in \R$, $\boldsymbol q \in \Srat^{D-1}$, it follows that \Cref{eq:frequencyav2} holds for almost all $(\boldsymbol k,f)\in \R^D\times R$.
\end{theorem}

\section{Ergodicity and Asymptotic Abelianness in a space-like cone from the Lieb-Robinson Bound} \label{section:ergodicity_abelianness}

First we discuss the Lieb-Robinson bound and asymptotic abelianness outside the light-cone defined by the Lieb-Robinson velocity. Then, we show that this implies clustering of correlations in factor invariant states, and hence space-like ergodicity. We state the Lieb-Robinson bound in the form of \cite[Corollary 4.3.3]{naaijkens_quantum_2017}:
\begin{lem}[Lieb-Robinson Bound] \label{lem:liebrobinsonbound}
Consider a dynamical system $(\mathfrak{U},ι,τ)$ (\Cref{defn:dynamicalsystem}) with interaction $Φ$ that satisfies Equation \eqref{eq:interaction}, that is
\begin{equation} 
    \norm{Φ}_λ \coloneqq \sup_{\boldsymbol n \in \Z^D} \sum_{X \ni \boldsymbol n} \norm{Φ(X)} |X| N^{2|X|} e^{λ \diam(X)} < \infty 
\end{equation}
for some $λ>0$.
Then, there exists a $υ_{LR}>0$ such that for all local $A,B \in \mathfrak{U}_{\rm loc}$ with supports $\supp(A)=Λ_A \in P_f( \mathbb{Z}^D)$, $\supp(B)=Λ_B \in P_f( \mathbb{Z}^D)$ respectively, and all $t \in \mathbb{R}$ we have the bound
\begin{equation}
    \norm{[ τ_t (A) , B]} \leq 4 \norm{A}\norm{B}|Λ_A| |Λ_B| N^{2 | Λ_A |} \exp{-λ( \dist ( A,B) - υ_{LR}|t|)} \label{eq:liebrobinsonbound}
\end{equation}
 with $υ_{LR}=2 \frac{\norm{Φ}_λ}{λ}$ called the Lieb-Robinson velocity.
\end{lem}

As mentioned, our results hold for all quantum spin lattice models with interactions that satisfy \Cref{eq:interaction}, including any finite range  and two-body exponentially decaying interactions.

The Lieb-Robinson bound implies that the algebra $\mathfrak{U}$ is asymptotically abelian under space-time translations along rays that are outside the light-cone defined by $υ_{LR}$. To make this precise, consider any lattice direction $\boldsymbol n \in \Z^D$, speed $υ>υ_{LR}$ and the subgroup of space-time translations along a ray $G_{υ,\boldsymbol n} \coloneqq  \{ (r\boldsymbol n, rυ^{-1}|\boldsymbol n|) : r \in \Z \} \subset \Z^D \times \R$. It is easy to prove the following:
\begin{theorem}[Space-like asymptotic abelianness] \label{th:asymptotic_abelianness}
Under the assumptions of Lemma \ref{lem:liebrobinsonbound}, the algebra $\mathfrak{U}$ is $G_{υ,\boldsymbol n}$-asymptotically abelian for every $υ \in \R$ s.t.\ $|υ|>υ_{LR}$, and $\boldsymbol n \in \Z^D$, i.e.\ 
\begin{equation} \label{eq:asympotic_abeliaN}
   \lim_{r \to \infty} \norm{ [ι_{r\boldsymbol n} τ_{rυ^{-1}|\boldsymbol{n}|}(A),B]} \rightarrow 0 
\end{equation}
for all $A,B \in \mathfrak{U}$.
\end{theorem}
\begin{proof}
First, consider $A,B\in \mathfrak{U}_{\rm loc}$ and let $Λ_A$, $Λ_B$ be the supports of $A$,$B$ respectively.  \Cref{eq:asympotic_abeliaN} is an immediate consequence of the Lieb-Robinson bound, indeed note that:
\begin{equation}
 \begin{array}{*3{>{\displaystyle}lc}p{5cm}}
    \dist(ι_{ \boldsymbol n}(A),B) = \dist(Λ_A+\boldsymbol n, Λ_B) &=& \min \{ |\boldsymbol y-\boldsymbol z| : \boldsymbol y \in Λ_A+\boldsymbol n, \boldsymbol z \in Λ_B \} \\
        &=& \min \{ |\boldsymbol y+\boldsymbol n-\boldsymbol z| : \boldsymbol y \in Λ_A, \boldsymbol z \in Λ_B \} \\
        &\geq& \min \{ |\boldsymbol n| - |\boldsymbol y-\boldsymbol z| :\boldsymbol y \in Λ_A, \boldsymbol z \in Λ_B \} \\
        &=& |\boldsymbol n| - \max\{ |\boldsymbol y-\boldsymbol z| : \boldsymbol y \in Λ_A ,\boldsymbol z \in Λ_B \} \\
        &\geq& |\boldsymbol n| - \diam(Λ_A \cup Λ_B)
    \end{array} 
\end{equation}
Using this in the Lieb-Robinson bound we have
\begin{equation}
 \begin{array}{*3{>{\displaystyle}lc}p{5cm}}
   \lefteqn{\norm{ [ι_{r\boldsymbol n} τ_{rυ^{-1}|\boldsymbol{n}|}(A),B]} \leq }\\ 
  &&4 \norm{A}\norm{B}|Λ_A| |Λ_B| N^{2 | Λ_A |} \exp{-λ( r |\boldsymbol n| + \diam (Λ_A \cup Λ_B) - r υ^{-1} |\boldsymbol n|υ_{LR})} 
\end{array}
\end{equation}
and the right hand-side of this inequality clearly vanishes at $r \to \infty$ whenever $υ>υ_{LR}$. 

Now consider arbritrary $A,B \in \mathfrak{U}$. Since $\mathfrak{U}$ is the norm completion of $\mathfrak{U}_{\rm loc}$ we can find Cauchy sequences $A_m, B_m \in \mathfrak{U}_{\rm loc}$ with $A_m \to A$ and $B_m \to B$. From the above derivation, we have that $\lim_{r \to \infty} \norm{ [ι_{r\boldsymbol n} τ_{rυ^{-1}|\boldsymbol{n}|}(A_m),B_m]} =0$, for all $m$. By uniform continuity of $ι$, $τ$ (as they are norm-preserving, $\norm{\iota_{r\boldsymbol n}(A)} = \norm{A}$, $\norm{\tau_t(A)} = \norm{A}$), and continuity of the commutator,
%as well as the fact that $ι$, $τ$ are automorphisms,
we can easily see that $\lim_{m} [ι_{r\boldsymbol n} τ_{rυ^{-1}|\boldsymbol{n}|}(A_m),B_m] = [ι_{r\boldsymbol n} τ_{rυ^{-1}|\boldsymbol{n}|}(A),B]$ uniformly for all $r$. This allows us to apply the Moore-Osgood  \Cref{th:MooreOsgood}, see \Cref{appendix:Moore-Osgood} and \cite[p. 140]{taylor_general_1985}, to get that the double limit and the two iterated limits exist and are equal 
\begin{equation}
 \begin{array}{*3{>{\displaystyle}lc}p{5cm}}
    \lim_m (\lim_r [ι_{r\boldsymbol n} τ_{rυ^{-1}|\boldsymbol{n}|}(A_m),B_m] )&=&
    \lim_r (\lim_m [ι_{r\boldsymbol n} τ_{rυ^{-1}|\boldsymbol{n}|}(A_m),B_m] ) \\
    &=& \lim_{m, r}([ι_{r\boldsymbol n} τ_{rυ^{-1}|\boldsymbol{n}|}(A_m),B_m] )
    \end{array}
    \end{equation} Hence, we have asymptotic abelianness for all $A,B \in \mathfrak{U}$.
\end{proof}
It is well established that given asymptotic abelianness of a C$^*$-algebra  under a certain group action $τ_g$, it follows that any factor state $ω$ satisfies clustering properties, in the sense that there is a net of elements $g_i \in G$ such that $\lim_i |ω(A τ_{g_i}(B) ) - ω(A) ω(B)| = 0$. This is shown in \cite[Theorem 8]{kastler_invariant_1966}, \cite[Example 4.3.24]{bratteli_operator_1987} (and in the latter referred to as ``strong mixing"). The most common example is space translations; any factor state is clustering in space. From these results, we obtain the following:

\begin{theorem}[Space-like clustering] \label{th:space-like-clustering}
Consider the assumptions Lemma \ref{lem:liebrobinsonbound} and a \textbf{factor invariant state} $ω$ of $\mathfrak{U}$. It follows that for every $\boldsymbol n \in \Z^D$, $|υ|>υ_{LR}$ and all observables $A,B \in \mathfrak{U}$ we have \textbf{space-like clustering of correlations}:
\begin{equation}
    \lim_{r \to \infty} | ω\big( ι_{r \boldsymbol n} τ_{rυ^{-1}|\boldsymbol n|} (A) B \big) - ω(A) ω(B) | = 0 .
\end{equation}
This immediately implies that $ω$ is also space-like ergodic (\Cref{def:spacelike-ergodic}) and space-like non-oscillating (\Cref{def:phase-spacelike-ergodic}).
\end{theorem}
%%%%%%%%%%%%%%%%%%%%%%%%%%%%%%%%%%%%%%%%%%%%%%%%%%%%%%%%%%%%%%%%%%%%%%%%%
\iffalse
\begin{proof}
The only difference with the proofs cited is that we have asymptotic abelianness only for the local observables. We have, by \cite[Example 4.3.24]{bratteli_operator_1987} that for any $ε>0$ there exists a finite family of $B_i \in \mathfrak{U}$, $i=1, \cdots , n$ and $A_i \in π_ω(\mathfrak{U})^{\prime}$ such that
\begin{equation}
  | ω\big( ι_{r \boldsymbol n} τ_{rυ^{-1}|\boldsymbol n|} (A) B \big) - ω(A) ω(B) | < ε + \sum_{i=1}^n \norm{A_i} \norm{ [ι_{r \boldsymbol n} τ_{rυ^{-1}|\boldsymbol n|} A, B_i]}
\end{equation}
We can now use the expression of the Lieb-Robinson bound for local $A$ and arbitrary $B \in \mathfrak{U}$, shown in \cite[Corollary 4.3.2]{naaijkens_quantum_2017}. This will give  $\lim_{r \to \infty} \norm{ [ι_{r \boldsymbol n} τ_{rυ^{-1}|\boldsymbol n|} A, B_i]} = 0$ for all $ε>0$, which concludes the proof.
\end{proof}
\fi
%%%%%%%%%%%%%%%%%%%%%%%%%%%%%%%%%%%%%%%%%%%%%%%%%%%%%%%%%%%%%%%%%%%%%%%%

The last statement in the theorem is shown as follows. It is trivial to see that $\lim_r| ω\big( ι_{r \boldsymbol n} τ_{rυ^{-1}|\boldsymbol n|} (A) B \big) - ω(A) ω(B) |=0$ implies  $$\lim_N \frac{1}{N} \sum_{r=1}^N  \big( ω\big( ι_{r \boldsymbol n} τ_{rυ^{-1}|\boldsymbol n|} (A) B \big) - ω(A) ω(B) \big) =0,$$ i.e.\ $ω$ is space-like ergodic. It is also easy to see that space-like clustering implies that the state is space-like non-oscillating, as in \Cref{def:phase-spacelike-ergodic}. To see this, first note that  $$\lim_N \frac{1}{N}\sum_{r=1}^N e^{-i(\boldsymbol k \cdot \boldsymbol n - f υ^{-1} |\boldsymbol n|)r}  \big( ω\big( ι_{r \boldsymbol n} τ_{rυ^{-1}|\boldsymbol n|} (A) B \big) - ω(A) ω(B) \big)=0,$$ from space-like clustering, and second that $$\lim_N \frac{1}{N}| \sum_{r=1}^N e^{-i(\boldsymbol k \cdot \boldsymbol n - f υ^{-1} |\boldsymbol n|)r} ω(A) ω(B) | =0$$ whenever $\boldsymbol k \cdot \boldsymbol n - f υ^{-1} |\boldsymbol n| \not\in 2\pi \Z$. 

A more intuitive condition, than that of factor states, is space clustering of correlations. A  state $ω$ is called space clustering if $ |ω( A_n B) - ω(A_n)ω(B)| \rightarrow  0 $ for any $B \in \mathfrak{U}$ and any sequence $A_n \in \mathfrak{U}$, $n=1, 2, \cdots $, with uniformly bounded norm, such that $\dist(A_n, B) \rightarrow 0$. It can be seen that space clustering invariant states are also space-like clustering and therefore space-like ergodic. This is shown in \Cref{appendix2}.

\begin{theorem} \label{th:space_like_clustering_from_spatial}
   Consider a dynamical system $(\mathfrak{U}, ι, τ)$ with interaction that satisfies \Cref{eq:interaction}, and a space clustering state $ω$, in the sense that $ ω( A_n B) - ω(A_n)ω(B) \rightarrow  0 $ for any $B \in \mathfrak{U}$, and any sequence $A_n \in \mathfrak{U}$, $n=1, 2, \cdots $  with uniformly bounded norm such that $\dist(A_n, B) \rightarrow 0$. It follows that $ω$ is space-like clustering, i.e.\ for all $υ>υ_{LR}$ and all $A,B \in \mathfrak{U}_{\rm loc}$ 
    \begin{equation}
        \lim_{\boldsymbol n \to \infty}  |ω \big( ι_{\boldsymbol n}τ_{υ^{-1}|\boldsymbol n|} (A) B \big) - ω(A) ω(B)|= 0.
    \end{equation}
This immediately implies that $ω$ is also space-like ergodic (\Cref{def:spacelike-ergodic}) and space-like non-oscillating (\Cref{def:phase-spacelike-ergodic})
\end{theorem}
This condition is, for example, satisfied in high temperature KMS states \cite[Theorem 3.2]{frohlich_properties_2015} where the connected correlation decays exponentially fast with the distance. In particular, in one spatial dimension we have exponential space clustering in all non-zero temperature KMS states by the result of Araki, \cite{araki_gibbs_1969}.

\section{Proofs of Ergodicity, Projection and Mean-Square Theorems} \label{sectproofergodicityprojection}

\subsection{Ergodicity and Projection Theorems \ref{th:maintheorem}, \ref{th:rankone}} \label{section:ergodicityproof}
%\Cref{th:frequencyproj1}
Consider a space-like ergodic $ω \in E_{\mathfrak{U}}$ and its respective GNS representation $(H_ω, π_ω, Ω_ω)$. Since $ω$ is $ι,τ$-invariant there exists a representation of $\mathbb{Z}^D \times \R$ by unitary operators $U(\boldsymbol n,t)$ acting on $H_ω$, see \Cref{prop:gns}. This representation is fully determined by the properties
\begin{equation}
    U_{\omega}(\boldsymbol n,t) \pi_{\omega}(A) U_{\omega}(\boldsymbol n,t)^{*}=\pi_{\omega}\left(ι_{\boldsymbol n} τ_t (A)\right)
\end{equation}
 and invariance of the cyclic vector $Ω_ω$
 \begin{equation}
     U_{\omega}(\boldsymbol n,t) \Omega_{\omega}=\Omega_{\omega}
 \end{equation}
for all $A \in \mathfrak{U}$, $(\boldsymbol n,t) \in \Z^D \times \R$.

Let $P_{υ,\boldsymbol n}$ be the orthogonal projection on the subspace of $H_ω$ that is invariant under $U_\omega(\boldsymbol n, υ^{-1}|\boldsymbol n|)$. This subspace will contain at least $Ω_ω$. We will first prove  Theorem \ref{th:rankone} using the following Lemmata.  % show that for almost all $υ$, $P_{υ,x}$ is the rank one projection onto $Ω_ω$, i.e.\ it is $|Ω_ω \rangle \langle Ω_ω|$, for arbitrary $x \in \mathbb{Z}^D$.

\begin{lem}
Consider the assumptions of Theorem \ref{th:rankone} and let $P_{υ,\boldsymbol n}^{(r)}$, $r= 1,2,...$, denote the projection onto the subspace of $H_ω$ spanned by vectors invariant under $U_ω^r (\boldsymbol n, υ^{-1}|\boldsymbol n|)$, with $P_{υ,\boldsymbol n}^{(1)}= P_{υ,\boldsymbol n}$. It follows that  $P_{υ,\boldsymbol n}^{(r)}$ is the rank one projection onto $Ω_ω$ for any $υ \in \hat{\mathbb{R}}$ with $|υ|>υ_c$, and all $\boldsymbol n \in \mathbb{Z}^D$, $r=1,2,...$ . 
\label{lem:1}
\end{lem}

\begin{proof}
Space-like ergodicity, expressed in the GNS representation (\Cref{prop:gns}) implies, for $A,B\in\mathfrak U_{\rm loc}$:
\[\arraycolsep=1.4pt\def\arraystretch{2.5}
\begin{array}{*3{>{\displaystyle}lc}p{5cm}}
\lim_{N \to \infty} \frac{1}{N} \sum_{m=0} ^{N-1} ω\big(  ι_{\boldsymbol n}^{rm} τ_{υ^{-1}|\boldsymbol n|}^{rm} (A) B \big)= ω(A)ω(B) \implies \\ %ω( ι_{x}^{n} τ_{ υ^{-1} |x|}^{n} ab )= ω(A)ω(B) 
\lim_{N \to \infty} \frac{1}{N} \sum_{m=0} ^{N-1} \langle Ω_ω ,  π_ω(ι_{\boldsymbol n}^{rm} τ_{υ^{-1}|\boldsymbol n|}^{rm} (A)B) Ω_ω \rangle = \langle Ω_ω, π_ω(A)Ω_ω \rangle \langle Ω_ω, π_ω(B)Ω_ω \rangle \implies \\
\lim_{N \to \infty} \frac{1}{N} \sum_{m=0} ^{N-1} \langle Ω_ω, U^{rm}_ω(\boldsymbol n, υ^{-1}|\boldsymbol n|)π_ω(A) \big(U^{rm}_ω(\boldsymbol n, υ^{-1}|\boldsymbol n|)\big)^* π_ω(B) Ω_ω \rangle = \\
=\langle Ω_ω, π_ω(A)Ω_ω \rangle \langle Ω_ω, π_ω(B)Ω_ω \rangle. 
\end{array} \]
Taking advantage of the invariance of $Ω_ω$ we get:
\begin{equation}
 \begin{array}{*3{>{\displaystyle}lc}p{5cm}}
\lim_{N \to \infty} \frac{1}{N} \sum_{m=0} ^{N-1} \langle Ω_ω,π_ω(A) \big(U^{rm}_ω(\boldsymbol n, υ^{-1}|\boldsymbol n|)\big)^* π_ω(B) Ω_ω \rangle = \\
=\langle Ω_ω, π_ω(A)Ω_ω \rangle \langle Ω_ω, π_ω(B)Ω_ω \rangle. 
\end{array} 
\end{equation}

Then we move the limit of the sum inside the inner product and the adjoint operation by  continuity and use von Neumann's mean ergodic \Cref{th:neumann} (\cite[Theorem II.11]{reed_i_1981}) which states that this limit will be the projection $P^{(r)}_{υ,\boldsymbol n}$ (in the strong operator topology). Thus, we get
 \begin{equation}
 \langle Ω_ω, π_ω(A) P^{(r)}_{υ, \boldsymbol n} π_ω(B) Ω_ω \rangle = \langle Ω_ω, π_ω(A)Ω_ω \rangle \langle Ω_ω, π_ω(B)Ω_ω \rangle \ , \ \ \forall A,B \in \mathfrak{U}_{\rm loc}.
\end{equation}
 By continuity of the inner product and $π_ω$ this will hold for all $A,B \in \mathfrak{U}$. By cyclicity of $Ω_ω$ the set $\operatorname{span} \{π_ω(B)Ω_ω : B \in \mathfrak{U} \}$ is dense in $H_ω$. Hence, this proves the Lemma.
%maybe explanation in between the lines
\end{proof}

\begin{lem}
Consider the assumptions of Theorem \ref{th:rankone}  and let $υ,w \in \hat{\mathbb{R}}$ with $υ \neq w$ and consider the projections $P_{υ,\boldsymbol n}$, $P_{w,\boldsymbol n}$ for some arbitrary $\boldsymbol n \in \mathbb{Z}^D$. If $Ψ \in P_{υ,\boldsymbol n}H_ω$, $Ψ \bot Ω_ω$ and $Ψ^{\prime} \in P_{w,\boldsymbol n}H_ω$, $Ψ^{\prime} \bot Ω_ω$, then $\langle Ψ,Ψ^{\prime}\rangle =0$. \label{lem:orthogonal}
\end{lem}

\begin{proof}
If $P_{υ,\boldsymbol n}$ (or $P_{w,\boldsymbol n}$) is rank one, then obviously $Ψ$ (resp. $Ψ^{\prime}$) will be $0$ and we are done. Suppose that $Ψ,Ψ^{\prime} \neq 0$ and without loss of generality $υ \neq 0$, $w \neq \infty$. 

\textit{Case 1: $w = 0$}. Then $Ψ^{\prime}$ invariant under time translations, so $U_ω(0,r)Ψ^{\prime}=Ψ^{\prime}$ for any $r \in \mathbb{Z}$. Hence, choosing $r=p|\boldsymbol n| \in \mathbb{Z}$ (note that $|\boldsymbol n|\in \N$ for the $\ell_1$ norm) for any arbitrary $p \in \mathbb{Z}$ we can write
\begin{equation}
\langle Ψ, Ψ^{\prime} \rangle = \langle Ψ, U_ω(0,p|\boldsymbol n|) Ψ^{\prime} \rangle   = \langle U_ω(0,-p|\boldsymbol n|)Ψ, Ψ^{\prime} \rangle .
\end{equation}
 By assumption $U_\omega(\boldsymbol n, υ^{-1} |\boldsymbol n|)Ψ=Ψ$, hence also  $U_\omega^q(\boldsymbol n, υ^{-1} |\boldsymbol n|)Ψ=Ψ$ for any integer $q$. Using this, and the fact that $[ U_ω(\boldsymbol n,t) , U_ω(\boldsymbol n ^{\prime},t^{\prime})]=0 , \forall \boldsymbol n,\boldsymbol n ^{\prime} \in \mathbb{Z}^D$ and $ t,t^{\prime} \in \mathbb{R}$, we have for any $q \in \mathbb{Z}$
\begin{equation}
\langle Ψ, Ψ^{\prime} \rangle =  \langle U_ω^q(\boldsymbol n, υ^{-1}|\boldsymbol n|)U_ω(0,-p|\boldsymbol n|)Ψ, Ψ^{\prime} \rangle = \langle U_ω^q\big(\boldsymbol n, (υ^{-1}- \frac{p}{q})  |\boldsymbol n| \big)Ψ, Ψ^{\prime} \rangle 
\end{equation}
where in the last equality we used the group properties.

For any $υ^{-1} \in \mathbb{R}$, we can find $p \in \mathbb{Z}$ and $q \in \mathbb{Z}_+$ such that \begin{equation}
|qυ^{-1} -p| < qυ_c ^{-1}.
\end{equation}
This is true for any $q$ sufficiently large so that $qυ^{-1}_c >1$, and $p = \floor{qυ^{-1}}$.
Thus, by choosing such $p,q$ we have:
\begin{equation}
    \langle Ψ, Ψ^{\prime} \rangle  = \langle U_ω^q(\boldsymbol n, z^{-1}|\boldsymbol n|)Ψ, Ψ^{\prime} \rangle 
\end{equation}
for $z^{-1} = υ^{-1} - p/q$ with $p,q$ such that $|z| > υ_c$. We can repeat the same process for $\langle U_ω^q(\boldsymbol n, z^{-1}|\boldsymbol n|)Ψ, Ψ^{\prime} \rangle$ to get:
\begin{equation}
 \begin{array}{*3{>{\displaystyle}lc}p{5cm}}
   \langle Ψ, Ψ^{\prime} \rangle&=&   \langle U_ω^q(\boldsymbol n, z^{-1}|\boldsymbol n|)Ψ, Ψ^{\prime} \rangle = \\ 
   &=&  \langle \big(U_ω^q(\boldsymbol n, z^{-1}|\boldsymbol n|) \big)^2Ψ, Ψ^{\prime} \rangle = ...=  \langle \big(U_ω^q(\boldsymbol n, z^{-1}|\boldsymbol n|)\big)^NΨ, Ψ^{\prime} \rangle 
 \end{array}
\end{equation}
for any $ N \in \mathbb{N}$. We then split  $\langle Ψ, Ψ^{\prime} \rangle$ as $\frac{1}{N}$ times a sum of $N$ equal terms:
\begin{equation}
     \langle Ψ, Ψ^{\prime} \rangle = \frac{1}{N} \sum_{m=0}^{N-1}  \langle \big(U_ω^q(\boldsymbol n, z^{-1}|\boldsymbol n|)\big)^mΨ, Ψ^{\prime} \rangle \ , \forall N \in \mathbb{N}.
     \end{equation}
By von Neumann's ergodic theorem the sum converges to the projection $P^{(q)}_{z,\boldsymbol n}$. Since $|z|>υ_c$ Lemma \ref{lem:1} gives that $P^{(q)}_{z,\boldsymbol n}$ is the rank one projection on $Ω_ω$ and by assumption $Ψ \bot Ω_ω$, hence 
\begin{equation}
     \langle Ψ, Ψ^{\prime} \rangle = \langle P_{z,\boldsymbol n}^{(q)}Ψ, Ψ^{\prime} \rangle = 0.
\end{equation}
\textit{Case 2: $w \neq 0$}. For any real $η>0$ we can find $p,q \in \mathbb{Z}$ so that $p/q \in [wυ^{-1}, w(υ^{-1} +η)]$, since the rationals are dense in $\mathbb{R}$. Choose $0<ε < \frac{ |1- wυ^{-1}|}{υ_c + |w|}$ such that \begin{equation}
w(υ^{-1}+ε) = \frac{p}{q} \in \mathbb{Q}. \label{eq:wυ}
\end{equation}
Taking advantage of invariance of $Ψ$ under $U_ω^q(\boldsymbol n,υ^{-1}|\boldsymbol n|)=U_ω(q\boldsymbol n,qυ^{-1}|\boldsymbol n|)$ and $Ψ^{\prime}$ under $U_ω^p(\boldsymbol n,w^{-1}|\boldsymbol n|)$ and the group property, like in the previous case, we can write:
\begin{equation}
\begin{array}{*3{>{\displaystyle}lcl}p{5cm}}
\langle Ψ, Ψ^{\prime} \rangle &=& \langle U_ω^q(\boldsymbol n,υ^{-1}|\boldsymbol n|)Ψ, Ψ^{\prime} \rangle \\
&= &\langle U_ω(0,-qε|\boldsymbol n|)U_ω(q\boldsymbol n,q(υ^{-1}+ε)|\boldsymbol n|)Ψ, Ψ^{\prime} \rangle .
\end{array}\end{equation}
Now, using Eq. \ref{eq:wυ} we have
\begin{equation}
\arraycolsep=1.4pt\def\arraystretch{1.3}
\begin{array}{*3{>{\displaystyle}lc}p{5cm}}
    \langle Ψ, Ψ^{\prime} \rangle&=&  \langle U_ω\big(0,-qε|\boldsymbol n|\big)U_ω\big(q\boldsymbol n,pw^{-1}|\boldsymbol n|\big) Ψ, Ψ^{\prime} \rangle \\
   &=&  \langle U_ω\big((q-p)\boldsymbol n,0\big)U_ω\big(0,-qε|\boldsymbol n|\big)U_ω\big(p\boldsymbol n,pw^{-1}|\boldsymbol n|\big) Ψ, Ψ^{\prime} \rangle \\
  &=&   \langle U_ω\big((q-p)\boldsymbol n,0\big)U_ω\big(0,-qε|\boldsymbol n|\big) Ψ, U_ω^{-p}\big(\boldsymbol n,w^{-1}|\boldsymbol n|\big)Ψ^{\prime} \rangle \\
   &=&  \langle U_ω\big((q-p)\boldsymbol n,0 \big)U_ω\big(0,-qε|\boldsymbol n|\big) Ψ, Ψ^{\prime} \rangle
    \end{array}
\end{equation}

where in the last equality we used invariance of $Ψ^{\prime}$. We can now define a $z \in \mathbb{R}$ by $-qε= z^{-1}(q-p)$, that is
\begin{equation}
    z= - \frac{q-p}{qε} = - \frac{1}{ε} + \frac{p}{q} \frac{1}{ε} =   - \frac{1}{ε}(1- \frac{w}{υ}) + w
\end{equation}
where in the last equality we used Eq. \ref{eq:wυ}. By the restrictions on $ε$, imposed above Eq. \ref{eq:wυ}, we can see $|z|>υ_c$. Hence, we can write
\begin{equation}
    \langle Ψ, Ψ^{\prime} \rangle = \langle U_ω^{q-p}(\boldsymbol n,  z^{-1}|\boldsymbol n|)Ψ, Ψ^{\prime} \rangle
\end{equation}
and similarly to the previous case we can repeat this process to get
\begin{equation}
\langle Ψ, Ψ^{\prime} \rangle = \langle P^{q-p}_{z,\boldsymbol n}Ψ, Ψ^{\prime} \rangle = 0. 
\end{equation}
\end{proof}
Equipped with this Lemma we can prove the Projection Theorem:
\begin{proof}[Proof of \Cref{th:rankone} ]
The quasi-local algebras of  quantum spin lattices are uniformly hyperfinite, or UHF, see \cite{glimm_certain_1960} . The GNS representation of a UHF algebra is over separable Hilbert space, since it is cyclic. The proof of this claim is exactly the same as \cite[Theorem 3.5]{glimm_certain_1960}, by using the cyclic vector $Ω_ω$. Thus every orthonormal set in $H_ω$ is countable.

Consider for  arbitrary $\boldsymbol n \in \mathbb{Z}^D$ the set
\begin{equation}
    \mathcal{K}_{\boldsymbol n} =\{ υ \in  \hat{\mathbb{R}}: {\rm rank}\,{P_{υ,\boldsymbol n}}>1. \}  
\end{equation}
For each  $υ \in \mathcal{K}_{\boldsymbol n}$ we can choose (by the axiom of choice) a non-zero $Ψ_{υ} \in P_{υ,\boldsymbol n}H_ω$ with $\norm{Ψ_υ} =1$ and (since ${\rm rank}\,{P_{υ,\boldsymbol n}}>1$) $Ψ_{υ} \bot Ω_ω$. By Lemma \ref{lem:orthogonal} we have $\langle Ψ_υ , Ψ_w \rangle=0$ for every $υ,w \in \mathcal{K}_{\boldsymbol n}$ with $υ \neq w$. Hence the cardinality of $\mathcal{K}_{\boldsymbol n}$ is the cardinality of an orthonormal set of $H_ω$, which is countable, i.e.\ of (Lebesgue) measure $0$. 

\iffalse %%%%%%%%%%%%%%%%%%%%%%%%%%%%%%%
Now consider any velocity of rational direction $υ = υ \frac{\boldsymbol n}{|\boldsymbol n|}$, $υ \in \R$, $\boldsymbol n \in \R$, and the limit as $T \to \infty$ of the operator
\begin{equation}
    \frac{1}{T} \int_0^T U_ω( \floor{υ \frac{\boldsymbol n}{|\boldsymbol n|} t}, t) \,dt  = \frac{1}{υT |\boldsymbol n|^{-1}} \int_0^{υT |\boldsymbol n|^{-1}} U_ω( \floor{x \boldsymbol n}, υ^{-1} |\boldsymbol n|x)  \,dx
\end{equation}
The limit can be taken over the integers, as the integrand is bounded and hence the integral is continuous in $T$.
\begin{equation}
   P \coloneqq \lim_{T \to \infty} \frac{1}{T} \int_0^T U_ω( \floor{υ \frac{\boldsymbol n}{|\boldsymbol n|} t}, t) \,dt = \lim_{N \to \infty} \frac{1}{N} \int_0^N U_ω( \floor{x \boldsymbol n}, υ^{-1} |\boldsymbol n|x)  \,dx
\end{equation}
We split the integral into a sum  into a sum $\int_0^1 \,dx_0 + \int_1^2 \,dx_1 + ...+ \int_{N-1}^N \,dx_{N-1}$ and in each of these integrals we change variable to $x_k=y_k+k$, with k being the lower limit of each integral:
\begin{equation}
    P= \lim_{N \to \infty} \frac{1}{N} \sum_{k=0}^{N-1} \int_0^1 U_ω( \floor{(y+k)\boldsymbol n}, (y+k)υ^{-1}|\boldsymbol n|) \,dy 
\end{equation}
now each $k n_i$ is an integer and can be moved outside the floor function. Using this fact and the group properties we get
\begin{equation}
      P= \lim_{N \to \infty} \frac{1}{N} \sum_{k=0}^{N-1} \int_0^1 U_ω( k \boldsymbol n, kυ^{-1} |\boldsymbol n| ) U_ω( \floor{y \boldsymbol n}, y υ^{-1})
\end{equation}
\fi
\end{proof}

Finally we  use von Neumann's ergodic theorem and the Projection Theorem \ref{th:rankone} to prove the Ergodicity Theorem \ref{th:maintheorem}.

\begin{proof}[Proof of Theorem \ref{th:maintheorem}]
Let $\boldsymbol q = \boldsymbol{n} / |\boldsymbol n| \in \Srat^{D-1}$ for $\boldsymbol n=(n_1,n_2,...,n_D) \in \mathbb{Z}^D$, and $A,B \in \mathfrak{U}$. For $υ>0$ (this suffices since we show the result for all $\boldsymbol n \in \Z^D$) and $\boldsymbol{υ}= υ \boldsymbol{q}=υ \frac{\boldsymbol n}{|\boldsymbol n|}$ we write:
\begin{equation} %\arraycolsep=1.4pt\def\arraystretch{2.2}
I \coloneqq  \frac{1}{T} \int_0^T ω \big( ι_{\floor{ \boldsymbol{υ} t}}τ_t (A) B \big)  \,dt 
= \frac{1}{υT| \boldsymbol n|^{-1}} \int_0^{υT| \boldsymbol n|^{-1}}  ω\big( ι_{\floor{x \boldsymbol n}} τ_{υ^{-1}|\boldsymbol n|x} (A) B\big) \,dx
\end{equation}
where we changed variable to $x=υt/|\boldsymbol n|$ and use the notation $\floor{x\boldsymbol n} = (\floor{x n_1}, ... , \floor{xn_D})$. We are interested in the limit $\lim_{T \to \infty} I $. Τhis limit can be taken over the integers since the integrand is bounded and hence the integral is continuous in $T$. We now write this expression in the GNS representation:
\begin{equation}
      I =  \frac{1}{N} \int_0^N \langle Ω_ω, U_ω\big(\floor{x \boldsymbol n},υ^{-1}|\boldsymbol n|x\big) π_ω(A)  U_ω^*\big(\floor{x \boldsymbol n},υ^{-1}|\boldsymbol n|x\big)  π_ω(B) Ω_ω \rangle \,dx
\end{equation}
\iffalse
\begin{equation}
\arraycolsep=1.4pt\def\arraystretch{2.2}
\begin{array}{*3{>{\displaystyle}lc}p{5cm}}
      I =  \lim_{ N \to \infty} \frac{1}{N} \int_0^N&& \langle Ω_ω, U_ω\big(\floor{x \boldsymbol n},υ^{-1}|\boldsymbol n|x\big) π_ω(A) \times  \\
  &&U_ω^*\big(\floor{x \boldsymbol n},υ^{-1}|\boldsymbol n|x\big)  π_ω(B) Ω_ω \rangle \,dx
    \end{array}
\end{equation}
\fi
We use invariance of $Ω_ω$ so that the integrand becomes:  $$\langle Ω_ω, π_ω(A) U_ω^*\big(\floor{x \boldsymbol n},υ^{-1}|\boldsymbol n|x\big) π_ω(B) Ω_ω \rangle.$$ We then split the integral into a sum $\int_0^1 \,dx_0 + \int_1^2 \,dx_1 + ...+ \int_{N-1}^N \,dx_{N-1}$ and in each of these integrals we change variable to $x_k=y_k+k$, with k being the lower limit of each integral:
%{\setstretch{2.4}
\begin{equation}
    I = \frac{1}{N} \sum_{k=0}^{N-1} \int_0^1 \langle Ω_ω, π_ω(A) U_ω^*\big(\floor{(y_k+k)\boldsymbol n},(y_k+k) υ^{-1}|\boldsymbol n|\big) \times \\
    π_ω(B) Ω_ω \rangle \, dy_k . \label{eq:splittingintegral}
\end{equation}%}
Now $kn_i$ is an integer and can be moved outside the floor function, and using the group properties we write:
\begin{equation}
\arraycolsep=1.4pt\def\arraystretch{2.7}
\begin{array}{*3{>{\displaystyle}lc}p{5cm}}  
    I &=&\frac{1}{N} \sum_{k=0}^{N-1} \int_0^1 \langle Ω_ω, π_ω(A) U^*_ω\big(k \boldsymbol n,kυ^{-1}|\boldsymbol n|\big) U^*_ω\big( \floor{y \boldsymbol n}, yυ^{-1}|\boldsymbol n|\big) π_ω(B) Ω_ω \rangle \,dy \\
    &=&\frac{1}{N} \sum_{k=0}^{N-1} \int_0^1 \langle Ω_ω, π_ω(A) \bigg(U^k_ω\big( \boldsymbol n,υ^{-1}|\boldsymbol n|\big) \bigg)^* U^*_ω\big( \floor{y \boldsymbol n}, yυ^{-1}|\boldsymbol n|\big) π_ω(B) Ω_ω \rangle \,dy
   \end{array}
\end{equation}
We now apply the large $N$ limit. We can calculate $\lim_{N \to \infty}I$ by using the bounded convergence theorem to exchange the limit of the sum and the integral,
 \begin{equation} 
 \begin{array}{*3{>{\displaystyle}lc}p{5cm}} 
 \lim_{N \to \infty} I= \\
     \int_0^1 \langle Ω_ω, π_ω(A) U_ω^* (\floor{y \boldsymbol n}, yυ^{-1}|\boldsymbol n|) \lim_{N \to \infty} \frac{1}{N} \sum_{n=0}^{N-1} \big(U^*_ω(\boldsymbol n,υ^{-1}|\boldsymbol n|)\big)^n  π_ω(B) Ω_ω \rangle\,dy
     \end{array}
 \end{equation}
and  using von Neumann's ergodic theorem, \cite[Theorem II.11]{reed_i_1981}:
\begin{equation}
    I=\int_0^1 \langle Ω_ω, π_ω(A) U_ω^* \big(\floor{y \boldsymbol n}, yυ^{-1}|\boldsymbol n|)  P_{υ,\boldsymbol n} π_ω(B) Ω_ω \rangle\,dy.
\end{equation}
Theorem \ref{th:rankone} tells us that $P_{υ,\boldsymbol n}$ is the rank one projection on $Ω_ω$ for almost all $υ$, hence:
\begin{equation}
    I= \int_0^1 \langle Ω_ω, π_ω(A) U_ω^* \big(\floor{y \boldsymbol n}, yυ^{-1}|\boldsymbol n|\big) Ω_ω\rangle \langle Ω_ω, π_ω(B) Ω_ω  \rangle \,dy \ \ \text{ a.e.} 
\end{equation}
and finally, since $Ω_ω$ is invariant we get
\begin{equation}
    I = \langle Ω_ω, π_ω(A)  Ω_ω\rangle \langle Ω_ω, π_ω(B) Ω_ω  \rangle = ω(A) ω(B) 
\end{equation}
for almost all $υ$. 
\end{proof}

\subsection{Proof of Mean-Square Ergodicity Theorem \ref{th:meansquared} and \Cref{th:meanN}} \label{section:meansquaredproof}

\begin{proof}[Proof of \Cref{th:meansquared}]
Let $\boldsymbol n \in \mathbb{Z}^D$, $A,B \in \mathfrak{U}$, $υ >0$ so that $\boldsymbol{υ}= υ \boldsymbol {n}/|\boldsymbol n|$ and consider
\begin{equation}
    I(T,T^{\prime}) = \frac{1}{T^{\prime}T} \int_0^{T^{\prime}} dt_2\int_0^{T} dt_1\, ω\big(ι_{\floor{ \boldsymbol{υ} t_1}}τ_{t_1} (A)  ι_{\floor{\boldsymbol{υ}t_2}}τ_{t_2} (B)  \big) \  , \ \  T,T^{\prime}>0.
\end{equation}

We will apply the Moore-Osgood \Cref{th:MooreOsgood} \cite[p. 140]{taylor_general_1985} to show that the double limit $\lim_{T,T^{\prime} \to \infty}I(T,T^{\prime})$ is equal to the iterated one. The Theorem states that if the iterated limit $\lim_{T^{\prime}\to\infty}(\lim_{T\to\infty}\cdot)$ exists, and the limit $\lim_{T\to\infty}\cdot$ exists uniformly in $T'$, then the double limit exists and is equal to the iterated one. We can easily calculate the iterated limit $\lim_{T^{\prime} \to \infty}\big(\lim_{T \to \infty}  I(T,T^{\prime}) \big)$ by using Theorem \ref{th:maintheorem}:
\begin{equation}
 \begin{array}{*3{>{\displaystyle}lc}p{5cm}} 
   \lim_{T^{\prime} \to \infty}\big(\lim_{T\to \infty} \frac{1}{T^{\prime}T} \int_0^{T^{\prime}} dt_2\int_0^{T} dt_1\, ω\big(ι_{\floor{ \frac{υ}{|\boldsymbol n|} t_1}\boldsymbol n}τ_{t_1} (A)  ι_{\floor{ \frac{υ}{|\boldsymbol n|} t_2}\boldsymbol n}τ_{t_2} (B)  \big) \big)\\ =
   \lim_{T^{\prime} \to \infty} \frac1{T'}
   \int_0^{T'} dt_2\,
   ω(A) ω\big(ι_{\floor{ \frac{υ}{|\boldsymbol n|} t_2}\boldsymbol n}τ_{t_2} (B)  \big)
   =ω(A)\omega(B)
   \end{array}
\end{equation}
for almost all $υ$. This shows that $\lim_{T \to \infty}$ exists pointwise in $T'$. In order to use the Moore-Osgood Theorem it remains to show that the limit on $T$ is uniform in $T'$. Consider the $T \to \infty$ limit, by following the same steps as in the proof of \Cref{th:maintheorem} in \Cref{section:ergodicityproof}. Since the integrand is bounded $I(T,T^{\prime})$ is continuous in $T$, hence the continuous limit can be changed with a discrete one. We have for all $T^{\prime}$, as in \Cref{eq:splittingintegral}
\begin{equation}
    \begin{array}{*3{>{\displaystyle}lc}p{5cm}}
    I(N,T^{\prime}) \coloneqq \\
    \frac{1}{T^{\prime}} \int_0^{T^{\prime}}dt_2 \int_0^1 dy \langle Ω_ω , π(A) \big( \frac{1}{N} \sum_{k=0}^{N-1} U^*(k \boldsymbol n, k υ^{-1}|\boldsymbol n|) \big) U^*( \floor{y\boldsymbol n}, yυ^{-1} |\boldsymbol n|) \times \\ U(\floor{υt_2},t_2) π(B) Ω_ω \rangle = \\
    \frac{1}{T^{\prime}} \int_0^{T^{\prime}}dt_2 \int_0^1 dy \langle \big( \frac{1}{N} \sum_{k=0}^{N-1} U(k \boldsymbol n, k υ^{-1}|\boldsymbol n|) \big) π(A^*) Ω_ω ,   U^*( \floor{y\boldsymbol n}, yυ^{-1} |\boldsymbol n|) \times \\ U(\floor{υt_2},t_2) π(B) Ω_ω \rangle.
    \end{array}
\end{equation}
The limit is done by using the bounded convergence theorem and von Neumann's ergodic theorem:
\begin{equation}
 \begin{array}{*3{>{\displaystyle}lc}p{5cm}}
    I(\infty,T^{\prime}) \coloneqq \lim_{N \to \infty} I(N,T') =\\
    =  \frac{1}{T^{\prime}} \int_0^{T^{\prime}}dt_2 \int_0^1 dy \langle  P_{υ,\boldsymbol n}π(A^*)Ω_ω , U^*( \floor{y\boldsymbol n}, yυ^{-1} |\boldsymbol n|) U(\floor{υt_2},t_2) π(B) Ω_ω \rangle  
\end{array}
\end{equation}
We can see
\begin{equation} \label{eq:Mean_squared_proof1}
  \begin{array}{*3{>{\displaystyle}lc}p{5cm}}
    |I(N,T^{\prime}) - I(\infty, T^{\prime})|= \\
  \bigg|   \frac{1}{T^{\prime}} \int_0^{T^{\prime}}dt_2 \int_0^1 dy \bigg\langle \big( \frac{1}{N} \sum_{k=0}^{N-1} U(k \boldsymbol n, k υ^{-1}|\boldsymbol n|) - P_{υ,\boldsymbol n}\big) π(A^*)Ω_ω ,  \\ U^*( \floor{y\boldsymbol n}, yυ^{-1} |\boldsymbol n|) U(\floor{υt_2},t_2) π(B) Ω_ω \bigg\rangle \bigg  |.
    \end{array}
\end{equation}
Using the Cauchy-Shwarz inequality and the fact that unitary operators preserve the norm of vectors,
\[
 \begin{array}{*3{>{\displaystyle}lc}p{5cm}}
    |I(N,T^{\prime}) - I(\infty, T^{\prime})| \leq  \\
   \leq \frac{1}{T^{\prime}}  \int_0^{T^{\prime}}dt_2 \int_0^1 dy \norm{\big( \frac{1}{N} \sum_{k=0}^{N-1} U(k \boldsymbol n, k υ^{-1}|\boldsymbol n|) - P_{υ,\boldsymbol n}\big) π(A^*)Ω_ω}\times\\  \norm{ U^*( \floor{y\boldsymbol n}, yυ^{-1} |\boldsymbol n|) U(\floor{υt_2},t_2) π(B) Ω_ω}\\ 
  = \norm{\big( \frac{1}{N} \sum_{k=0}^{N-1} U(k \boldsymbol n, k υ^{-1}|\boldsymbol n|) - P_{υ,\boldsymbol n}\big) π(A^*)Ω_ω}\,
   \norm{π(B) Ω_ω}.
    \end{array}
\]
Finally, for any $Ψ \in H_ω$ and for any  $ε>0$ there exists $N_0$ such that $N\geq N_0$ implies $\norm{ \big(\frac{1}{N} \sum_{k=0}^{N-1} U^*(k \boldsymbol n, k υ^{-1}|\boldsymbol n|) - P_{υ,\boldsymbol n} \big)Ψ} <ε$, by von Neumann's ergodic theorem. Therefore
\begin{equation}
    |I(T,T^{\prime}) - I(\infty, T^{\prime})| \leq ε \norm{π(B)Ω_ω} \ \text{ for all } T^{\prime}
\end{equation}
which shows that the limit is uniform in $T^{\prime}$.
Hence, the double limit exists and is equal to the iterated one, which completes the proof.
\end{proof}

\begin{proof}[Proof of \Cref{th:meanN}]
For simplicity and clarity of the notation, we do the case where all times are equal $T_1=T_2=\ldots=T_n$, and all operators are the same $A_1=A_2=\ldots=A_n$. The extension to the general case and the existence of the multiple limit will be discussed afterwards.

The proof is  done inductively.  Suppose  we have for $n \in \mathbb{N}$ that 
\begin{equation}
   \lim_{T \to \infty} \frac{1}{T^n} ω \bigg( \big(\int_0^T  ι_{\floor{ \boldsymbol{υ} t}}τ_{t} A  \,dt  \big)^n \bigg)= \big( ω(A) \big)^n \ , \text{ a.e. in υ}
\end{equation}
which is certainly true for $n=2$ by the Mean-Square Ergodicity Theorem. 
We will show that the same will then hold for $n+1$. In order to do so, we first note that, using the GNS representation, the Cauchy-Schwartz inequality, boundedness of the state and the $C^*$ property,
\begin{equation}\label{eq:basicineq}
    |\omega(AB)|^2 \leq \omega(AA^*)\omega(B^*B)
    \leq \norm{A}^2\omega(B^*B).
\end{equation}
We consider $  I_{n+1}(T) \coloneqq \frac{1}{T^{n+1}} ω \bigg( \big(\int_0^T  ι_{\floor{ \boldsymbol{υ} t}}τ_{t} A  \,dt  \big)^{n+1} \bigg) $ in the form
\begin{equation}
\arraycolsep=1.4pt\def\arraystretch{2.4}
  \begin{array}{*3{>{\displaystyle}lc}p{5cm}}
 I_{n+1}(T) &=&\frac{1}{T^{n+1}} ω \bigg( \prod_{j=1}^{n+1} \int_0^T ι_{\floor{\boldsymbol{υ} t_j}} τ_{t_j} A \, dt_j \bigg)\\
 &=& \frac{1}{T^{n+1}} \Big(\prod_{j=1}^{n+1}\int_0^Tdt_j\Big)\,  ω \bigg( \prod_{j=1}^{n+1} ι_{\floor{\boldsymbol{υ} t_j}} τ_{t_j} A \bigg).
 \end{array}
\end{equation}
We will show that the distance of $I_{n+1}(T)$ from $ω(A)I_n(T) $ goes to $0$ in the $T \to \infty$ limit.  We denote $\int^{(n)} = \Big(\prod_{j=1}^{n}\int_0^T dt_j\Big)$ and write:
\begin{equation} \label{eq:moment_theorem_proof_1}
     \begin{array}{*3{>{\displaystyle}l}p{5cm}}
 Δ_{n+1}(T) \coloneqq |I_{n+1(T)} - \omega(A)I_n(T)| \\
 = \big| \frac{1}{T^{n+1}} \int^{(n)} \int_0^T dt_{n+1} \,\omega\big( \prod_{j=1}^{n} ι_{\floor{ \boldsymbol{υ} t_j}}τ_{t_j} (A)\, (ι_{\floor{ \boldsymbol{υ} t_{n+1}}}τ_{t_{n+1}}(A)  - ω(A)) \big) \big |\\
 \leq 
 \frac{1}{T^{n+1}} \int^{(n)} \,\big| \omega\big( \prod_{j=1}^{n} ι_{\floor{ \boldsymbol{υ} t_j}}τ_{t_j} (A)\, \int_0^T dt_{n+1}(ι_{\floor{ \boldsymbol{υ} t_{n+1}}}τ_{t_{n+1}}(A)  - ω(A)) \big) \big |.
    \end{array}
\end{equation} 
By submultiplicativity of the operator norm and the fact that $\iota$ and $\tau$ preserve the norm,
\begin{equation}
    \norm{\prod_{j=1}^{n} ι_{\floor{ \boldsymbol{υ} t_j}}τ_{t_j} (A)}
    \leq \norm{A}^n.
\end{equation}
Therefore, from \Cref{eq:basicineq} and performing the integral $\int^{(n)}$,
\begin{equation} 
\arraycolsep=1.4pt\def\arraystretch{2.5}
     \begin{array}{*3{>{\displaystyle}l}p{5cm}}
    \Delta_{N+1}(T)^2\leq \\
    \leq\norm{A}^{2n}\frac1{T^2}
    \int_0^T dt
    \int_0^T dt'\,
    \omega\bigg(
    \big(ι_{\floor{ \boldsymbol{υ} t}}τ_{t}(A) - ω(A)\big)^*
    \big(ι_{\floor{ \boldsymbol{υ} t'}}τ_{t'}(A) - ω(A)\big)\bigg)\nonumber\\
    =
    \norm{A}^{2n}\bigg(\frac1{T^2}
    \int_0^T dt\int_0^T dt'\,
    \omega\big(
    ι_{\floor{ \boldsymbol{υ} t}}τ_{t}(A^*)ι_{\floor{ \boldsymbol{υ} t'}}τ_{t'}(A)\big) - |ω(A)|^2\bigg)
\end{array}
\end{equation}
where in the last step we used invariance of the state. We now use the  Mean-Square Ergodicity Theorem \ref{th:meansquared}, which shows that for almost every $υ$, i.e.\  $ Δ_{n+1}(T) \to 0$. This means that the limit of $I_{n+1}(T)$ coincides with the limit of
\begin{equation}
    ω(A) \frac{1}{T^N} \int_0^T \langle Ω_ω, \prod_{j=1}^{n} π_ω\big(ι_{\floor{ \boldsymbol{υ} t_j}}τ_{t_j} A \big) Ω_ω \rangle \,dt_j
\end{equation}
which is $ω(A) ω(A)^n= ω(A)^{n+1}$, for almost all $υ$, by the induction hypothesis.

In the case where the times $T_1,T_2,\ldots$ are different, the induction hypothesis states the existence of the multiple limit for $I_n(T_1,T_2,\ldots,T_n)$, which is certainly true for $n=2$ by the Mean-Square Ergodicity Theorem. Then, we note that the integral $\int^{(n)}$ was performed before taking the limit on the last time $T_{n+1}$. Therefore, we find that $\Delta_{n+1}(T_1,T_2,\ldots,T_{n+1})$ is bounded, uniformly on all $T_1,T_2,\ldots T_n$, by a value that vanishes as $T_{n+1}\to\infty$. Thus the limit $\lim_{T_{n+1}\to\infty} I_{n+1}(T_1,T_2,\ldots,T_{n+1})$ exists uniformly on $T_1,T_2,\ldots,T_n$. Since the iterated limit exists (on $T_{n+1}\to\infty$, then on $T_1,T_2,\ldots,T_n\to\infty$) by the induction hypothesis, we can apply the Moore-Osgood theorem, and the multiple limit on $I_{n+1}(T_1,T_2,\ldots,T_{n+1})$ exists and gives the stated result. Note that the Moore-Osgood \Cref{th:MooreOsgood} holds for double limits over the Cartesian product of arbitrary directed  sets $X \times Y$, see \Cref{appendix:Moore-Osgood}.

The case of different $A_1,A_2,\ldots$ is obtained from an immediate generalisation of the steps above.

\end{proof}

\section{Proof of Strong Operator Convergence in GNS Space of KMS States \Cref{th:KMS_average_strong_convergence}} \label{section:strong_convergence}
In this Section we prove \Cref{th:KMS_average_strong_convergence} and from it we subsequently derive \Cref{th:mean_asymptotic_gns_abelianness}. To do this we first show the following Lemmata. The following Lemma is also used in proving the extended ergodicity theorem \Cref{th:general_theorem} inductively.
\begin{lem} \label{lemma:6.1}
Consider the assumptions of \Cref{th:KMS_average_strong_convergence}. It follows that for any $A,B,C,D \in \mathfrak{U}$, any rational direction $\boldsymbol{q} \in \Srat^{D-1}$ of the velocity $\boldsymbol{υ}= υ \boldsymbol q$ and almost every speed  $υ \in \mathbb{R}$ 
\begin{equation} \label{eq:lemma6.11}
    \lim_{T \to \infty} \frac{1}{T} \int_0^T  ω_β( B ι_{\floor{ \boldsymbol{υ} t}}τ_t (A) D) \, dt = ω_β(A) ω_β(B \, D)
\end{equation}
and
\begin{equation} \label{eq:lemma6.12}
     \lim_{T \to \infty} \frac{1}{T^2} \int_0^T \int_0^T ω_β( B ι_{\floor{ \boldsymbol{υ} t}}τ_t (A) \, ι_{\floor{ \boldsymbol{υ} {t^{\prime}}}}τ_{t^{\prime}} (C) D) \, dt \, dt^{\prime} = ω_β(A)ω_β(C) ω_β(B \, D)
\end{equation}
\end{lem}

\begin{proof}
We first show \Cref{eq:lemma6.11} for arbitrary $A,B,D \in \mathfrak{U}$. The idea is to use the Ergodicity \Cref{th:maintheorem}, but to do this we have to change the order of elements inside $ω_β( \cdots)$. This is done by exploiting the KMS condition, which however does not necessarily hold for all elements of the algebra, but for a dense subset. Hence, we consider a sequence of analytic elements $B_n \in \mathfrak{U}_τ$ that satisfies the KMS condition, see \Cref{defn:KMS}, and approximates $B$, i.e.\ $B_n \to B$. To show \Cref{eq:lemma6.11} we have to calculate the iterated limit
\begin{equation} \label{proof:lemm611_1}
     \lim_{T \to \infty} \frac{1}{T} \int_0^T  ω_β( \lim_n B_n ι_{\floor{ \boldsymbol{υ} t}}τ_t (A) D) \, dt
\end{equation}
but to apply the KMS condition we need to do the $\lim_n$ after the $\lim_T$. For this we will use again the Moore-Osgood \Cref{th:MooreOsgood} \cite[p. 140]{taylor_general_1985} to show that the double limit exist (and is equal to the iterated ones).
Consider $I_n(T)=\frac{1}{T} \int_0^T  ω_β( B_n ι_{\floor{ \boldsymbol{υ} t}}τ_t (A) D) $. Using the KMS condition for $B_n$ the integrand becomes:
\begin{equation}
     ω_β( B_n ι_{\floor{ \boldsymbol{υ} t}}τ_t (A) D) =  ω_β(  ι_{\floor{ \boldsymbol{υ} t}}τ_t (A) D \tau'_{iβ} B_n ) .
\end{equation}
Applying the average $\frac{1}{T}\int_0^T \,dt$  and taking the limit $T \to \infty$ we get, by using the Ergodicity \Cref{th:maintheorem}:
\begin{equation}
 \lim_{T \to \infty}    \frac{1}{T}\int_0^T  ω_β( B_n ι_{\floor{ \boldsymbol{υ} t}}τ_t (A) D)  \,dt = ω_β(A) ω_β(D \tau'_{iβ}B_n)
\end{equation}
for almost every $υ$. Re-applying the KMS condition (in the opposite direction) we get  that $\lim_{T \to \infty} I_n(T)$ exists for all $n$:
\begin{equation}
    \lim_{T \to \infty} \frac{1}{T}\int_0^T  ω_β( B_n ι_{\floor{ \boldsymbol{υ} t}}τ_t (A) D)  \,dt = ω_β(A) ω_β(B_n D) \ , \ \ \forall n \in \N.
\end{equation}

It is easy to see that  $\lim_n I_n(T)$ exists and is $\frac{1}{T}\int_0^T  ω_β( B ι_{\floor{ \boldsymbol{υ} t}}τ_t (A) D)  \,dt$ \textbf{uniformly} in $T$, for instance by using the dominated convergence theorem to bring the limit inside the integral with the uniform bound $|ω_β( (B_n-B) ι_{\floor{ \boldsymbol{υ} t}}τ_t (A) D)|\leq \norm{B_n-B}\norm{A}\norm{D}$, then the Cauchy-Schwartz inequality, and finally the uniform bound $ω_β( D^*ι_{\floor{ \boldsymbol{υ} t}}τ_t (A^*A) D)\leq \norm{D}^2\norm{A}^2$.  By the Moore-Osgood theorem this means that the double limit and the two iterated limits exist and are equal. Hence the result holds for any $A,B,D \in \mathfrak{U}$ in the sense of the iterated limit:
\begin{equation}
      \lim_{T \to \infty} \frac{1}{T}\int_0^T  ω_β( \lim_nB_n ι_{\floor{ \boldsymbol{υ} t}}τ_t (A) D)  \,dt = ω_β(A) ω_β(B D) .
\end{equation}

The proof of \Cref{eq:lemma6.12} is similar. We consider again arbitrary $A,B,C,D \in \mathfrak{U}$ and sequences $A_n, B_n \in \mathfrak{U}_{\tau'}$ of analytic elements for $\tau'$ that satisfy the KMS condition and $A_n \to A$, $B_n \to B$. We first use the KMS condition to write:
\begin{equation}
     ω_β\big( B_n ι_{\floor{ \boldsymbol{υ} t}}τ_t (A_n) \, ι_{\floor{ \boldsymbol{υ} {t^{\prime}}}}τ_{t^{\prime}} (C) D\big) =  ω_β\big(  ι_{\floor{ \boldsymbol{υ} t}}τ_t (A_n) \, ι_{\floor{ \boldsymbol{υ} {t^{\prime}}}}τ_{t^{\prime}} (C) D \tau'_{iβ}(B_n)\big)
\end{equation}
apply the average $\frac{1}{T^2} \int_0^T \int_0^T \,dt \,dt^{\prime}$ and consider $\lim_{T \to \infty}$.  To calculate this limit we consider the double limit: $$\lim_{T,T^{\prime} \to \infty} \frac{1}{T^{\prime}T} \int_0^T \int_0^{T^{\prime}} ω_β(  ι_{\floor{ \boldsymbol{υ} t}}τ_t (A_n) \, ι_{\floor{ \boldsymbol{υ} {t^{\prime}}}}τ_{t^{\prime}} (C) D \tau'_{iβ}(B_n)) \, dt \, dt^{\prime}.$$ Define:
\begin{equation} \label{eq:lemma_6.1_double_limit}
    I_n(T,T^{\prime}) \coloneqq \frac{1}{T^{\prime}T} \int_0^T \int_0^{T^{\prime}} ω_β\big(  ι_{\floor{ \boldsymbol{υ} t}}τ_t (A_n) \, ι_{\floor{ \boldsymbol{υ} {t^{\prime}}}}τ_{t^{\prime}} (C) D \tau'_{iβ}(B_n)\big) \, dt \, dt^{\prime}.
\end{equation}
We intend to apply the Moore-Osgood \Cref{th:MooreOsgood}. First consider $\lim_{T}$  of $I_n(T,T^{\prime}) $. Using the dominated convergence Theorem to move the $\lim_{T}$ inside the integral, and then the Ergodicity \cref{th:maintheorem} we see that 
\begin{equation} \label{eq:lemma611_proof_limit}
    \lim_{T \to \infty} I_n(T,T^{\prime}) = \frac{1}{T^{\prime}} \int_0^{T^{\prime}} ω_β(A_n) ω_β\big( ι_{\floor{ \boldsymbol{υ} {t^{\prime}}}}τ_{t^{\prime}} (C) D \tau'_{iβ}(B_n)\big) \,dt^{\prime}.
\end{equation}
In exactly the same manner as in the proof of the Mean-Square Ergodicity theorem, in \Cref{section:meansquaredproof}, and specifically in the derivation around \Cref{eq:Mean_squared_proof1}, we can see that this limit is also uniform in $T^{\prime}$. Now, the $\lim_{T^{\prime}}$ can be done by using the KMS condition to write 
\begin{equation}
    I_n(T,T^{\prime}) =  \frac{1}{T^{\prime}T} \int_0^T \int_0^{T^{\prime}} ω_β\big( \, ι_{\floor{ \boldsymbol{υ} {t^{\prime}}}}τ_{t^{\prime}} (C) D \tau'_{iβ}(B_n)   \tau'_{iβ} ι_{\floor{ \boldsymbol{υ} t}}τ_t (A_n)\big) \, dt \, dt^{\prime}
\end{equation}
and again using dominated convergence and the ergodicity theorem we see that the limit exists and is
\begin{equation}
    \lim_{T^{\prime} \to \infty} I_n(T,T^{\prime}) = \frac{1}{T} \int_0^T ω_β(C) ω_β\big( \, ι_{\floor{ \boldsymbol{υ} {t^{\prime}}}}τ_{t^{\prime}} (C) D \tau'_{iβ}(B_n)   \tau'_{iβ} ι_{\floor{ \boldsymbol{υ} t}}τ_t (A_n) \big) \,dt.
\end{equation}

Hence, by the Moore-Osgood theorem the double limit exists and is equal to the iterated one. This implies that also the limit $\lim_{T \to \infty} I_n(T,T)$ exist and is equal to the double limit. The iterated limit is the one that we can calculate, by reapplying the ergodicity theorem in \Cref{eq:lemma611_proof_limit} 
\begin{equation}
    \lim_{T \to \infty} I_n(T,T) = ω_β(A_n) ω_β(C) ω_β\big(D \tau'_{iβ} (B_n)\big) =ω_β(A_n) ω_β(C) ω_β( B_n D ) 
\end{equation}
where in the last line we re-applied the KMS condition for $B_n$. Now, in a similar manner as before, we can see that $\lim_n I_n (T,T)$ exists uniformly in $T$. Hence, we can apply to Moore-Osgood \Cref{th:MooreOsgood} to get the result.
\end{proof}

\begin{lem} \label{lemma:6.2}
Consider the assumptions of \Cref{th:KMS_average_strong_convergence}. Then for any $A \in \mathfrak{U}$ it follows that for all $B \in \mathfrak{U}$, any rational direction $\boldsymbol{q} \in \Srat^{D-1}$ and almost every speed  $υ \in \mathbb{R}:$
\begin{equation}
\arraycolsep=1.4pt\def\arraystretch{1.8}
\begin{array}{*3{>{\displaystyle}lc}p{5cm}}
  \lim_{T \to \infty} \frac{1}{T^2}  \int_0^T \int_0^T ω_β \big( B^* \big(\,ι_{\floor{ \boldsymbol{υ} t}}τ_t (A)-ω_β(A) \mathds{1} \big)^* \big(\,ι_{\floor{ \boldsymbol{υ} {t^{\prime}}}}τ_{t^{\prime}} (A)-ω_β(A) \mathds{1} \big)  \,B \big) \,dt \,dt^{\prime} \\ = 0 
    \end{array}
\end{equation}
with $\boldsymbol{υ}= υ \boldsymbol q$.
\end{lem}

\begin{proof}
The integrand is the sum of four terms:
\begin{equation}
\arraycolsep=1.4pt\def\arraystretch{1.2}
\begin{array}{*3{>{\displaystyle}lc}p{5cm}}
    ω_β(B^* \, ι_{\floor{ \boldsymbol{υ} t}}τ_t (A^*)\, ι_{\floor{ \boldsymbol{υ} {t^{\prime}}}}τ_{t^{\prime}} (A) \, B ) - ω_β(B^*  ι_{\floor{ \boldsymbol{υ} t}}τ_t (A^*)\,  B )ω_β(A) - \\ ω_β(A) ω_β( B^* ι_{\floor{ \boldsymbol{υ} {t^{\prime}}}}τ_{t^{\prime}} (A) B) + ω_β(A^*)ω_β(A) ω_β(B^*B)
    \end{array}
\end{equation}
The result easily follows by applying \Cref{lemma:6.1} for each of these terms.
\end{proof}

Finally we can prove \Cref{th:KMS_average_strong_convergence}:
\begin{proof}[Proof of \Cref{th:KMS_average_strong_convergence}]
We have by \Cref{lemma:6.2} that  for all $A,B \in \mathfrak{U}$ and almost every speed $υ$:
\begin{equation}
\begin{array}{*3{>{\displaystyle}lc}p{5cm}}
    \lim_{T \to \infty} \frac{1}{T^2}  \int_0^T \int_0^T ω_β \big(B^*&&\big(\,ι_{\floor{ \boldsymbol{υ} t}}τ_t (A)-ω_β(A) \mathds{1} \big)^* \times \\
    &&\big(\,ι_{\floor{ \boldsymbol{υ} {t^{\prime}}}}τ_{t^{\prime}} (A)-ω_β(A) \mathds{1} \big)  \,B \big) \,dt \,dt^{\prime}  = 0 .
    \end{array}
\end{equation}
We write this in the GNS representation of $ω_β$:
\begin{equation}
\begin{array}{*3{>{\displaystyle}lc}p{5cm}}
    \lim_{T \to \infty} \frac{1}{T^2}  \int_0^T \int_0^T  \big\langle Ω_{ω_β}, && π_{ω_β} \big( B^* \big(\,ι_{\floor{ \boldsymbol{υ} t}}τ_t (A)-ω_β(A) \mathds{1} \big)^* \times \\
    &&\big(\,ι_{\floor{ \boldsymbol{υ} {t^{\prime}}}}τ_{t^{\prime}} (A)-ω_β(A) \mathds{1} \big)  \,B \big) Ω_{ω_β} \big\rangle \, dt \,dt^{\prime}=0
    \end{array}
\end{equation}
which can be written as
\begin{equation}
\arraycolsep=1.4pt\def\arraystretch{2.3}
\begin{array}{*3{>{\displaystyle}lc}p{5cm}}
     \lim_{T \to \infty}   \big\langle  &&\frac{1}{T}\int_0^T π_{ω_β}    \big(\,ι_{\floor{ \boldsymbol{υ} t}}τ_t (A)-ω_β(A) \mathds{1} \big) π_{ω_β}(B) Ω_{ω_β}  \, dt,  \\
    &&\frac{1}{T}\int_0^T π_{ω_β} \big(\,ι_{\floor{ \boldsymbol{υ} {t^{\prime}}}}τ_{t^{\prime}} (A)-ω_β(A) \mathds{1}  \big) π_{ω_β}(B) Ω_{ω_β} \, dt^{\prime} \big\rangle =0
     \end{array}
\end{equation}
where the integrals are to be understood as Bochner integrals (see \cite[Section 3.8]{hille_functional_1996}). This gives:
\begin{equation}
    \lim_{T \to \infty} \norm{ \frac{1}{T}\int_0^T π_{ω_β}  \big(\,ι_{\floor{ \boldsymbol{υ} t}}τ_t (A)-ω_β(A) \mathds{1} \big) π_{ω_β}(B)Ω_{ω_β}  \, dt} =0
\end{equation}
for any $A,B \in \mathfrak{U}$ and almost all $υ$. By cyclicity of $Ω_{ω_β}$ this statement reads
\[\arraycolsep=1.4pt\def\arraystretch{2.5}
\begin{array}{*3{>{\displaystyle}lc}p{5cm}}
      \lim_{T \to \infty} \norm{ \frac{1}{T}\int_0^T π_{ω_β} \big(  \,ι_{\floor{ \boldsymbol{υ} t}}τ_t (A)-ω_β(A) \mathds{1} \big) Ψ  \, dt} =0 \implies \\ \lim_{T \to \infty} \norm{  \frac{1}{T}\int_0^T π_{ω_β} \big(  \,ι_{\floor{ \boldsymbol{υ} t}}τ_t (A) \big) Ψ \, dt - ω_β(A) Ψ} =0 
      \end{array}
\]
for any $Ψ \in H_{ω_β}$, i.e.\ we have that $\frac{1}{T} \int_0^T π_{ω_β} \big(  \,ι_{\floor{ \boldsymbol{υ} t}}τ_t (A) \big) \, dt$ (understood as a  Bochner integral) converges in the Strong operator topology to $ω_β(A) \mathds{1}$, for almost every speed $υ$.

\end{proof}

\subsection{Proof of Mean Asymptotic Abelianness in  GNS space of KMS state \Cref{th:mean_asymptotic_gns_abelianness}}\label{section:mean_asympt_abel}

\Cref{th:mean_asymptotic_gns_abelianness} is an immediate corollary of \Cref{th:KMS_average_strong_convergence}:

\begin{proof}[Proof of \Cref{th:mean_asymptotic_gns_abelianness}]
This is an immediate consequence of \Cref{th:KMS_average_strong_convergence}. Indeed, consider the limit (in the norm of the Hilbert space):

\[\arraycolsep=1.4pt\def\arraystretch{2.2}
\begin{array}{*3{>{\displaystyle}lc}p{5cm}}
    &&\lim_{T \to \infty} \frac{1}{T} \int_0^T π_{ω_β} ( [ι_{\floor{ \boldsymbol{υ} t}}τ_t (B) ,C ] ) \,dt \,Ψ = \\
    &=& \lim_{T \to \infty} \frac{1}{T} \int_0^T \big(π_{ω_β}( ι_{\floor{ \boldsymbol{υ} t}}τ_t (B) ) π_{ω_β}(C) -   π_{ω_β}(C) π_{ω_β}( ι_{\floor{ \boldsymbol{υ} t}}τ_t (B) ) \big) \,dt \,Ψ \\
    &=& ω_β(B) π_{ω_β}(C)Ψ - π_{ω_β}(C) ω_β(B)Ψ =0
    \end{array}
    \]
where in the last line we used convergence (in the norm) of $$ \frac{1}{T} \int_0^T π_{ω_β} ( ι_{\floor{ \boldsymbol{υ} t}}τ_t (B) )\,dtΨ$$ to $ω_β(B)Ψ$ by \Cref{th:KMS_average_strong_convergence}.
\end{proof}

\section{Proof of Extended Almost Everywhere Ergodicity \Cref{th:general_theorem}} \label{section:proof_of_general_theorem}
\begin{proof}[Proof of \Cref{th:general_theorem}]
 We proceed to prove the Theorem by induction on $n$. 
For $n=1$ the Theorem holds by \Cref{lemma:6.1}. Suppose that the Theorem holds for arbitrary $n-1$

 \[
\begin{array}{*3{>{\displaystyle}lc}p{5cm}}
   \lim_{T_1,\ldots,T_{n-1} \to \infty}  ω \big(  A_1   \overline{B_1^{T_1}} A_2 \overline{B_2^{T_2}}  \cdots \overline{B_{n-1}^{T_{n-1}}}  A_{n} \big) = ω(B_1) \cdots ω(B_{n-1}) ω( A_1 \cdots A_n)
 
    \end{array}
    \]
we show that it will then also hold for $n$. Consider
\begin{equation}
    I(T_1,\ldots,T_{n-1},T_n) \coloneqq ω\big(    A_1   \overline{B_1^{T_1}} A_2 \overline{B_2^{T_2}}  \cdots \overline{B_{n-1}^{T_{n-1}}}  A_{n} \overline{B_n^{T_n}} A_{n+1}\big).
\end{equation}
The goal is to show that: $$\lim_{T_1,\ldots,T_n \to \infty} I(T_1,\ldots,T_{n-1},T_n)=  ω(B_1) \cdots ω(B_{n-1}) ω(B_n) ω(  A_1 \cdots A_n A_{n+1})$$ and we do this by showing that the double limit on $T_1,\ldots,T_{n-1}$ and $T_n$ exists, by the Moore-Osgood \Cref{th:MooreOsgood}. This is done very similarly to the proof of \Cref{eq:lemma6.12} of \Cref{lemma:6.1}, hence we will omit the technical details. 

By the induction hypothesis we have that the following limit exists
\begin{equation} \label{eq:extended_theorem_iterated_limit}
    \lim_{T_1,\ldots,T_{n-1} \to \infty} I(T_1,\ldots,T_{n-1},T_n) =  ω(B_1) \cdots ω(B_{n-1}) ω( A_1 \cdots A_n \overline{B_{n}^{T_n}} A_{n+1})
\end{equation}
for any $T_n$. 

Additionally, by \Cref{lemma:6.1} and in particular \Cref{eq:lemma6.11}, the following limit exists
\begin{equation}
    \lim_{T_n \to \infty} I(T_1,\ldots,T_{n-1},T_n) =  ω(B_n)ω\big(    A_1   \overline{B_1^{T_1}} A_2 \overline{B_2^{T_2}}  \cdots \overline{B_{n-1}^{T_{n-1}}}  A_{n}  A_{n+1}\big)
\end{equation}
and in fact it exists uniformly in $T_1,\ldots,T_{n-1}$. This is shown similarly to how the double limit is calculated after \Cref{eq:lemma_6.1_double_limit}. 

Therefore, by the Moore-Osgood Theorem, the double limit exists and is equal to the iterated ones: take the $T_n\to\infty$ limit of \Cref{eq:extended_theorem_iterated_limit} and apply \Cref{lemma:6.1}:
\begin{equation}
\arraycolsep=1.4pt\def\arraystretch{1.3}
\begin{array}{*3{>{\displaystyle}lc}p{5cm}}
    \lim_{T_1,\ldots,T_{n-1},T_n \to \infty}  I(T_1,\ldots,T_{n-1},T_n) = \lim_{T_n \to \infty } \lim_{T_1,\ldots,T_{n-1} \to \infty} I(T_1,\ldots,T_{n-1},T_n)\\
    = ω(B_1) \cdots ω(B_{n-1}) ω(B_n) ω(  A_1 \cdots A_n A_{n+1}).
    \end{array}
\end{equation}
\end{proof}

\section{Proofs of Ergodic Theorems Considering Oscillations} \label{section:frequencyproof}

%\subsection{All velocities of rational direction, almost all wavenumbers and frequencies}
\begin{proof}[Proof of Theorem \ref{th:frequencyav2}]
Consider $A,B \in \mathfrak{U}$, $\boldsymbol q = \frac{\boldsymbol n}{|\boldsymbol n|} \in \Srat ^{D-1}$ and velocity $\boldsymbol{υ} = υ \boldsymbol q$, $υ \in \mathbb{R}$. Consider any non-zero wavenumber-frequency pair $(\boldsymbol k, f) \in \mathbb{R}^D \times \mathbb{R}$.  We follow the same steps as in the proof of Theorem \ref{th:maintheorem} (see \Cref{section:ergodicityproof}):
\begin{equation}
\begin{array}{*3{>{\displaystyle}lc}p{5cm}}
  I=  \lim_{T \to \infty} \frac{1}{T} \int_0^T  e^{i (\boldsymbol k\cdot \boldsymbol v   -  f) t} \big( ω \big( ι_{\floor{ \boldsymbol{υ}t}}τ_t (A) B \big)-ω(A) ω(B) \big) \,dt  \\
  = \lim_{N \to \infty} \frac{1}{N} \int_0^{N}  e^{i(\boldsymbol k \cdot \boldsymbol n - fυ^{-1}|\boldsymbol n|)x}\big(ω\big( ι_{\floor{x \boldsymbol n}} τ_{υ^{-1}|\boldsymbol n|x} (A) B\big) - ω(A)ω(B) \big) \,dx
   \end{array}
\end{equation}
Firstly, the case $\boldsymbol k \cdot \boldsymbol{υ}-f=0$ is covered by the Ergodicity \Cref{th:maintheorem}. 
Going to the GNS representation and splitting the integral as in Eq. \ref{eq:splittingintegral} we get:
\begin{equation} \label{eq:frequencyavproof1}
\begin{array}{*3{>{\displaystyle}lc}p{5cm}}
    I= \lim_{N \to \infty} \frac{1}{N} \sum_{m=0}^{N-1} \int_0^1 e^{i(\boldsymbol k \cdot \boldsymbol n - f υ^{-1} |\boldsymbol n|)(y+m)} \bigg( \\    \big \langle Ω_ω, π_ω(A) U_ω^*( m\boldsymbol n,m υ^{-1}|\boldsymbol n|) U^*_ω ( \floor{y\boldsymbol n}, yυ^{-1}|\boldsymbol n|) π_ω(B) Ω_ω \big\rangle -ω(A)ω(B) \bigg)\,dy 
    \end{array}
\end{equation}

We split this as a sum of limits, since we will see that both exist. For the left term we can again use the bounded  convergence theorem to move the limit of the sum inside the integral. In this case, we have to use von Neumann's ergodic theorem for the limit
\begin{equation} \label{eq:frequencyprojectionlimit}
   \lim_N \frac{1}{N}\sum_{m=0}^{N-1} \big(e^{-i(\boldsymbol k \cdot \boldsymbol n - f υ^{-1} |\boldsymbol n|)}U_ω( \boldsymbol n, υ^{-1}|\boldsymbol n|) \big)^m
\end{equation}
which will give (in the strong operator topology)  a projection $P_{υ,\boldsymbol n}(\boldsymbol k, f)$  on the subspace of $H_ω$ formed by vectors invariant under  $e^{-i(\boldsymbol k \cdot \boldsymbol n - f υ^{-1} |\boldsymbol n|)}U_ω( \boldsymbol n, υ^{-1}|\boldsymbol n|)$. 
Hence, Eq. \ref{eq:frequencyavproof1} implies:
\begin{equation}\label{eq:frequencyavproof3}
\begin{array}{*3{>{\displaystyle}lc}p{5cm}}
    I = \int_0^1 e^{i(\boldsymbol k \cdot \boldsymbol n - fυ ^{-1}|\boldsymbol n|)y} \langle Ω_ω, π_ω(A) P_{υ,\boldsymbol n}(\boldsymbol k,f) U^*_ω( \floor{y \boldsymbol n}, yυ^{-1}|\boldsymbol n|) π_ω(B) Ω_ω \rangle \,dy  \\
   - \lim_N\frac{1}{N}\sum_{m=0}^{N-1} \int_0^1 e^{i(\boldsymbol k \cdot \boldsymbol n - fυ ^{-1}|\boldsymbol n|)(y+m)} ω(A)ω(B) \, dy
    \end{array}
\end{equation}
Note that when $\boldsymbol k \cdot \boldsymbol n - fυ ^{-1}|\boldsymbol n|$ is a multiple of $2π$, then the projection will be $P_{υ,\boldsymbol n}$, as in \Cref{th:rankone}, which is established to be rank one for almost all $υ$ and hence we immediately get 
\begin{equation}
    I= \int_0^1 e^{2πi \, y} ω(A) ω(B) \, dy - \lim_N \frac{1}{N} \sum_{m=0}^{N-1} \int_0^1 e^{2πi \, (y+m)}ω(A)ω(B) \, dy = 0
\end{equation}
We claim that  for all other $(\boldsymbol k,f)$ the projection $P_{υ,\boldsymbol{k}}(\boldsymbol k,f)$ will be $0$ for almost all $υ$, which would conclude the proof. We confirm this in what follows, by proving \Cref{th:frequencyproj}

%%%%%%%%%%%%%%%%%%%%%%%% ALMOST ALL FREQ. -- ALL V
\iffalse
This follows by discreteness of the point spectrum (the eigenvalues) of $U_ω( \boldsymbol n, υ^{-1}|\boldsymbol n|)$. If there exists  a non-zero invariant $Ψ \in H_ω$:
\begin{equation}
    e^{-i(\boldsymbol k \cdot \boldsymbol n - f υ^{-1} |\boldsymbol n|)}U_ω( \boldsymbol n, υ^{-1}|\boldsymbol n|) Ψ= Ψ \label{eq:frequencyavproof2}
\end{equation}
then $Ψ$ is an eigenvector of $U_ω( \boldsymbol n, υ^{-1}|\boldsymbol n|)$ with eigenvalue $e^{i(\boldsymbol k \cdot \boldsymbol n - f υ^{-1} |\boldsymbol n|)}$.  Since $U_ω(\boldsymbol n, υ^{-1}|\boldsymbol n|)$ has a countable number of eigenvalues, it follows that there are countably many pairs $(\boldsymbol k, f)$ s.t. $P_{υ,\boldsymbol n}(\boldsymbol k, f)$ is not $0$. For the $(\boldsymbol k,f)$ s.t. $e^{-i(\boldsymbol k \cdot \boldsymbol n - f υ^{-1} |\boldsymbol n|)}$ is not an eigenvalue, the requirement  Eq. \ref{eq:frequencyavproof2} is only satisfied by the $0$ vector, i.e.\ the projection is $0$. This also concludes the proof of the first part of \Cref{th:frequencyproj}.
\fi

%with the projection giving $I=0$ for almost all $(\boldsymbol k, f) \in \mathbb{R}^D \times \mathbb{R}$.
\end{proof}

%\subsection{All wavenumber-frequency pairs, almost all velocities of rational direction}
The proof of \Cref{th:frequencyproj} follows the same steps as the proof of  \Cref{th:rankone} in \Cref{section:ergodicityproof}. Using the same arguments as in the proofs of \Cref{lem:1} and \Cref{lem:orthogonal} we can show two similar Lemmata respectively:
\begin{lem} 
Consider the assumptions of Theorem \ref{th:frequencyproj} and let $P_{υ,\boldsymbol n}^{(r)}(\boldsymbol k,f)$, $r= 1,2,...$, denote the projection onto the subspace of $H_ω$ spanned by vectors invariant under $e^{-i(\boldsymbol k \cdot \boldsymbol n - fυ^{-1}|\boldsymbol n|)r}U_ω^r (\boldsymbol n, υ^{-1}|\boldsymbol n|)$, with $P_{υ,\boldsymbol n}^{(1)}(\boldsymbol k,f)= P_{υ,\boldsymbol n}(\boldsymbol k,f)$. It follows that the range of $P_{υ,\boldsymbol n}^{(r)}(\boldsymbol k,f)$ is $\{0\}$ for any $υ \in \hat{\mathbb{R}}$ with $|υ|>υ_c$, and all $\boldsymbol n \in \mathbb{Z}^D$, $r=1,2,...$.
\label{lem:3}
\end{lem}

\begin{lem} \label{lem:frequencyorthogonal}
Consider the assumptions of \Cref{th:frequencyproj} and  $(\boldsymbol k, f)\in \R^D \times \R$. Let  $υ,w \in \hat{\mathbb{R}}$ with $υ \neq w$ and consider the projections $P_{υ,\boldsymbol n}(\boldsymbol k,f)$, $P_{w,\boldsymbol n}(\boldsymbol k, f)$ for some arbitrary $\boldsymbol n \in \mathbb{Z}^D$. If $Ψ \in P_{υ,\boldsymbol n}(\boldsymbol k,f)H_ω $ and $Ψ^{\prime} \in P_{w,\boldsymbol n}(\boldsymbol k,f)H_ω$, then $\langle Ψ, Ψ^{\prime} \rangle =0$. 
\end{lem}

Then, using \cref{lem:frequencyorthogonal} and the countability of orthonormal sets of the GNS Hilbert space $H_ω$, we can conclude te proof of \Cref{th:frequencyproj}, i.e.\ that for any $(\boldsymbol k,f)$ the range of the projection $P_{υ,\boldsymbol n}(\boldsymbol k,f)$ is $\{ 0 \}$  for almost all $υ \in \mathbb{R}$ and all $\boldsymbol n \in \Z^D$. Using this in \Cref{eq:frequencyavproof3} also concludes the proof of \Cref{th:frequencyav2}. With that, the proof of \Cref{eq:frequency-moments} is done similarly to \Cref{section:meansquaredproof}

The proof of \Cref{th:frequencyaverage1} is much simpler, taking advantage of the countability of eigenvalues:
\begin{proof}[Proof of \Cref{th:frequencyaverage1}]
Under the same maniupulations as above, we arrive at  \Cref{eq:frequencyavproof1}. By von Neumann's ergodic theorem we hve to consider the projection given by the limit \ref{eq:frequencyprojectionlimit}, i.e.\ the projection onto vectors that are invariant under $e^{-i(\boldsymbol k \cdot \boldsymbol n - f υ^{-1} |\boldsymbol n|)}U_ω( \boldsymbol n, υ^{-1}|\boldsymbol n|)$. We will argue that this projection has range $\{0\}$ for almost all $(\boldsymbol k,f)$. Indeed, consider $(\boldsymbol k,f)$  with $\boldsymbol k \cdot \boldsymbol n - fυ^{-1}|\boldsymbol n| \neq 2πik$, $\forall k=0,1,2,\dots$, as the $(\boldsymbol k,f)$ that make this quantity a multiple of $2πi$ are of measure $0$ in $\R^D\times \R$. Then, the arguement follows from the  discreteness of the point spectrum (the eigenvalues) of $U_ω( \boldsymbol n, υ^{-1}|\boldsymbol n|)$. If there exists  a non-zero invariant $Ψ \in H_ω$:
\begin{equation}
    e^{-i(\boldsymbol k \cdot \boldsymbol n - f υ^{-1} |\boldsymbol n|)}U_ω( \boldsymbol n, υ^{-1}|\boldsymbol n|) Ψ= Ψ \label{eq:frequencyavproof2}
\end{equation}
then $Ψ$ is an eigenvector of $U_ω( \boldsymbol n, υ^{-1}|\boldsymbol n|)$ with eigenvalue $e^{i(\boldsymbol k \cdot \boldsymbol n - f υ^{-1} |\boldsymbol n|)}$.  Since $U_ω(\boldsymbol n, υ^{-1}|\boldsymbol n|)$ has a countable number of eigenvalues, it follows that there are countably many pairs $(\boldsymbol k, f)$ s.t. $P_{υ,\boldsymbol n}(\boldsymbol k, f)$ is not $0$. For the $(\boldsymbol k,f)$ s.t. $e^{-i(\boldsymbol k \cdot \boldsymbol n - f υ^{-1} |\boldsymbol n|)}$ is not an eigenvalue, the requirement  Eq. \ref{eq:frequencyavproof2} is only satisfied by the $0$ vector, i.e.\ the projection is $0$.
\end{proof}

\appendix

\section{Remark on integration} \label{appendix1}
Throughout this work we deal with ray averaged observables $\frac{1}{T} \int_0^T τ_t ι_{\floor{\boldsymbol{υ}t}} A \, \dd t$, where $A \in \mathfrak{U}$ and $(\mathfrak{U},ι,τ)$ is a dynamical system. This integral is to be understood in the Bochner sense \cite{hille_functional_1996}. It is easy to see that the function $τ_t ι_{\floor{\boldsymbol{υ}t}} A$ is Bochner integrable \cite[Definition 3.7.3]{hille_functional_1996} in any compact subset of $\R$,  as it can be approximated by the countable valued functions (simple functions):
\begin{equation}
    A_n(t) =τ_{\frac{\lfloor nt \rfloor }{n}}ι_{\lfloor \boldsymbol{υ} t \rfloor}  A
\end{equation}
 Indeed, $τ_t$ is strongly continuous and $\lim_n \frac{\lfloor nt \rfloor }{n}=t$. Hence, $$\lim_n \norm{A_n(t)-τ_t ι_{\floor{\boldsymbol{υ}t}} A}=0 $$ for all $t$, for any $A \in \mathfrak{U}$. The integral is defined in any interval $[0,T]$ of $\R$, uniquely and independently of the choice of $A_n(t)$ as
\begin{equation}
    \int_0^T τ_t ι_{\floor{\boldsymbol{υ}t}} A \, \dd t \coloneqq \lim_{n} \int_0^T    A_n(t) \, \dd t
\end{equation}
where the integral of $A_n(t)$ is defined as a sum over its (countable) values, in the usual way. The Bochner integral calculated inside a linear functional can be pulled outside, by \cite[Theorem 3.7.12]{hille_functional_1996}, in particular for any state $ω$ we have $ω\big(  \int_0^T τ_t ι_{\floor{\boldsymbol{υ}t}} A \, \dd t \big)= \int_0^T ω\big( τ_t ι_{\floor{\boldsymbol{υ}t}} A \big) \, \dd t$.

Likewise, in the GNS representation  $\int_0^T π\big(τ_t ι_{\floor{\boldsymbol{υ}t}} A  \big)\, \dd t$ is defined, and the two definitions are compatible in the sense that the integral commutes with $π$, i.e.\ $π(\int_0^T τ_t ι_{\floor{\boldsymbol{υ}t}} A  \, \dd t) = \int_0^T π\big(τ_t ι_{\floor{\boldsymbol{υ}t}} A \big)\, \dd t$,  by \cite[Theorem 3.7.12]{hille_functional_1996}. Finally, the integral $\int_0^T π\big(τ_t ι_{\floor{\boldsymbol{υ}t}} A  \big) Ψ \, \dd t$ is well defined and equal to  $\int_0^T π\big(τ_t ι_{\floor{\boldsymbol{υ}t}} A  \big) \, \dd t Ψ$, again by \cite[Theorem 3.7.12]{hille_functional_1996} and considering $ π\big(τ_t ι_{\floor{\boldsymbol{υ}t}} A  \big) Ψ$ as the evaluation map from $π(\mathfrak{U})$ to $H_ω$.

\section{Multiple Limits and the Moore-Osgood Theorem} \label{appendix:Moore-Osgood}

The Moore-Osgood Theorem is a general tool that lets us deal with double limits over arbitrary sets, given an appropriate limiting notion on these sets (called direction). We give an overview of  \cite[p. 139]{taylor_general_1985}. We assume a non-empty set $X$ and a family of subsets $\mathcal{N}$ of $X$ called a \textbf{direction on $X$ }such that $N_1, N_2 \in \mathcal{N}$ implies that there exists a $N_3 \in \mathcal{N}$ such that $N_3 \subset N_1 \cap N_2$. For example, in our Theorems we deal with  $\lim_{T \to \infty}$ which corresponds to the direction $\mathcal{T} = \big \{ \{ T \in \R :  T\geq m \}$, $m \in \N \big\}$.

A pair $(g,\mathcal{N})$ of a  function $g:X \to Z$, where $Z$ a metric space, and a direction $\mathcal{N}$ is called a directed function. We say that\textbf{ $(g,\mathcal{N})$  converges to $x_0 \in Z$, and write $(g,\mathcal{N}) \to x_0$}, if for each neighborhood $U$ of $x_0$ there exists a $N$ in $\mathcal{N}$ such that $g(N)\subset U$.  If $X,Y$ are sets with directions $\mathcal{M}, \mathcal{N}$, respectively, then the family of sets $\big\{M \times N$, $M \in \mathcal{M}$, $N \in \mathcal{N} \big\}$ is a direction in $X \times Y$, denoted $\mathcal{M \times N}$. The double limit of a function $f: X \times Y \to Z$, where $(Z,d)$ is a metric space, is given by the convergence of $(f , \mathcal{M} \times \mathcal{N})$. It is easy to see that $( f(x, \cdot) , \mathcal{N})$
and $(f(\cdot, y) , \mathcal{M})$ are directed functions. We say that $( f(x, \cdot) , \mathcal{N})$ converges to $g(x)$ \textbf{uniformly}, if for every $ε>0$ there exists a $N \in \mathcal{N}$ s.t. $d\big( f(x,y) , g(x) \big) < ε$ for $x\in X$ and $y \in N$.

The notion of the multiple limits $\lim_{T_1,\ldots,T_{n} \to \infty} I(T_1, T_2\ldots,T_n)$ in the main text is to be understood as convergence of the directed function $(I,\mathcal{T}^n)$. It is clear that the multiple limit $\lim_{T_1,\ldots,T_{n} \to \infty}$ coincides with the double limit $\lim_{(T_1,\ldots,T_{n-1}) , T_n \to \infty}$ over $\R^{n-1} \times \R$. One can then take advantage of the Moore-Osgood Theorem:
\begin{theorem}[Moore-Osgood] \label{th:MooreOsgood}
 Consider a complete metric space $Z$, the sets $X,Y$ with directions $\mathcal{M}, \mathcal{N}$ respectively  and a function $f: X \times Y \to Z$. Denote by $\mathcal{P}$  the direction $\mathcal{M} \times \mathcal{N}$. If there exist a function $g(x)$ s.t. $( f(x, \cdot), \mathcal{N}) \to g(x)$ uniformly on $x \in X$, and a function $h(y)$ s.t. $ ( f(\cdot,y) , \mathcal{M} ) \to h(y)$ for any $y \in Y$, then $(f,\mathcal{P})$ converges to a $z \in Z$ and $(g,\mathcal{M}) \to z$, $(h,\mathcal{N})\to z$.
\end{theorem}

\section{Space-like ergodicity under the assumption of spatial clustering} \label{appendix2}
Here we show that a space clustering state in combination with the Lieb-Robinson bound is space-like ergodic. In order to do this we approximate the time evolution of observables by local ones, by using the result \cite[Corollary 4.4]{nachtergaele_quasi-locality_2019}:
\begin{lem} \label{lem:localapprox}
Let $A\in \mathfrak{U}$ and consider a finite $Λ \subset \mathbb{Z}^D$. If there is an $ε>0$ such that 
\begin{equation}
\norm{ [ A, B] }\leq ε \norm{A}  \norm{B}  \ , \ \ \forall B \in \mathfrak{U}_{ \mathbb{Z}^D \setminus Λ}
\end{equation}
then we can approximate $A$ by a strictly local $\mathbb{P}_Λ (A) \in \mathfrak{U}_Λ$:
\begin{equation}
\norm{\mathbb{P}_Λ (A) - A } \leq 2ε \norm{A}
\end{equation} 
\end{lem}
 Combined with the Lieb-Robinson bound \Cref{lem:liebrobinsonbound} we can prove (this is also described in \cite[Chapter 4.3]{naaijkens_quantum_2017}):
\begin{prop} \label{prop:LRapproximation}
    Consider the assumptions of \Cref{lem:liebrobinsonbound}, the time evolution $τ_t A$ of a local $A \in \mathfrak{U}_Λ$ and the finite sets $Λ_r = \cup_{\boldsymbol x \in Λ}B_{\boldsymbol x}(r)$, $r=1,2,3, \dots$ extending a distance $r$ around $Λ$. Then, we can approximate $τ_t A$ by the local $\mathbb{P}_{Λ_r}(τ_t A) \in \mathfrak{U}_{Λ_r}$ :
    \begin{equation}
        \norm{ \mathbb{P}_{Λ_r} (τ_t(A)) - τ_t(A)} \leq 2ε_r \norm{A}
    \end{equation}
    with $ε_r = C |Λ|N^{2|Λ|} \exp{-λ(r- υ_{LR}|t|)}$, $r=1,2,3 \dots$.
\end{prop}

Using this result we can show:

\begin{proof}[Proof of \Cref{th:space_like_clustering_from_spatial}]
Consider $\boldsymbol n \in \Z^D$ and $A,B  \in \mathfrak{U}_{\rm loc}$ with
\begin{equation}
    \supp(A) \coloneqq Λ_A \ , \ \ \supp(B) \coloneqq Λ \text{ and } \supp(ι_{\boldsymbol n} A) \coloneqq Λ_A+\boldsymbol n
\end{equation}
Let $r \in \N$ and by \Cref{prop:LRapproximation} consider the local approximation of $τ_t B$ by $\mathbb{P}_{Λ_r}(τ_t B)$, supported in $Λ_r = \cup_{\boldsymbol x \in Λ} B_{\boldsymbol x}(r)$:
\begin{equation}
\norm{ \mathbb{P}_{Λ_r} (τ_t B) - τ_t B} \leq 2C |Λ| N^{2|Λ|} \norm{B} \exp{-λ(r-υ_{LR} |t|)}  \label{ineq:5.1localapproximationB}
\end{equation}
Let $ω$ be a time and space invariant pure state and consider the quantity $I \coloneqq | ω\big(ι_{\boldsymbol n}(A) \mathbb{P}_{Λ_r}(τ_t B)\big) - ω(A)ω(\mathbb{P}_{Λ_r}(τ_t B))|$, where we are interested in first applying the large $r$ limit, and then showing that the $\boldsymbol n \to \infty$ limit is 0.  By linearity and the triangle inequality, it holds for any $r,l \in \N$ :
\[\arraycolsep=1.4pt\def\arraystretch{2.2}
\begin{array}{*3{>{\displaystyle}lc}p{5cm}}
I  &=& \big| ω\bigg(ι_{\boldsymbol n}(A) \big(\mathbb{P}_{Λ_r}(τ_t B) + \mathbb{P}_{Λ_l}(τ_t B) -  \mathbb{P}_{Λ_l}(τ_t B) \big)\bigg) - ω(A)ω\big(\mathbb{P}_{Λ_r}(τ_t B)\\&+&\mathbb{P}_{Λ_l}(τ_t B)  - \mathbb{P}_{Λ_l}(τ_t B)\big) \big| \\
&\leq& \big| ω\big(ι_{\boldsymbol n}(A) \mathbb{P}_{Λ_l}(τ_t B) \big)-ω(A) ω(\mathbb{P}_{Λ_l}(τ_t B))\big|+  \\
&+&\big| ω\big(ι_{\boldsymbol n}(A) \big(\mathbb{P}_{Λ_r}(τ_t B)  - \mathbb{P}_{Λ_l}(τ_t B)\big) \big)-ω(A) ω\big(\mathbb{P}_{Λ_r}(τ_t B)  - \mathbb{P}_{Λ_l}(τ_t B)\big)\big|
\end{array} 
\]

The idea is to control the first part using  clustering and the second one using the approximation of time evolution. We now estimate the quantities $T_1=\big| ω\big(ι_{\boldsymbol n}(A) \mathbb{P}_{Λ_l}(τ_t B) \big)-ω(A) ω(\mathbb{P}_{Λ_l}(τ_t B))\big|$ and  $T_2=\big| ω\big(ι_{\boldsymbol n}(A) \big(\mathbb{P}_{Λ_r}(τ_t B)  - \mathbb{P}_{Λ_l}(τ_t B)\big) \big)-ω(A) ω\big(\mathbb{P}_{Λ_r}(τ_t B)  - \mathbb{P}_{Λ_l}(τ_t B)\big)\big|$. Consider $r$ large enough and 
\begin{equation}
    l= \floor{ε \dist (Λ_A+\boldsymbol n, Λ)} + \floor{ε \diam (Λ_A \cup Λ)} + 2 \ \text{ for some } 0<ε<1 \label{eq:5l}
\end{equation}

By clustering of correlations, we easily see that the term $T_1$ goes to $0$ as $\boldsymbol n \to \infty$ since the state is invariant and the distance between $A$ and $ι_{-\boldsymbol n} \mathbb{P}_{Λ_l}(τ_t B)$  goes to infinity, and $\norm{ι_{-\boldsymbol n} \mathbb{P}_{Λ_l}(τ_t B)}\leq \norm{B}$. We will see that this is also the case for $T_2$, in the large $r$, large $\boldsymbol n$ limit. To estimate $T_2$, we use the triangle inequality and note that $ω$ has norm $1$:
\begin{equation}
    T_2 \leq  \norm{ι_{\boldsymbol n}A} \norm{\mathbb{P}_{Λ_r}(τ_t B)  - \mathbb{P}_{Λ_l}(τ_t B)}  +  \norm{A} \norm{\mathbb{P}_{Λ_r}(τ_t B)  - \mathbb{P}_{Λ_l}(τ_t B)}  \label{ineq:5.3T2}
\end{equation}
where $ \norm{ι_{\boldsymbol n}A}=\norm{A}$. Using the approximation of $τ_tB$, inequality \ref{ineq:5.1localapproximationB}, we get
\begin{equation}
    T_2 \leq 2\norm{A}  2C |Λ| N^{2|Λ|} \norm{B} \big( \exp{-λ(r-υ_{LR} |t|)} +  \exp{-λ(l-υ_{LR} |t|)}\big)
\end{equation}

We consider a compact subset $T\coloneqq ε^{\prime} υ^{-1}_{LR} [-|\boldsymbol n|,|\boldsymbol n|]$ for some $0<ε^{\prime}<ε$, and let $t \in T$. Then, for the value \ref{eq:5l} of $l$, since $\floor{y} \geq y-1, \forall y>0$, we have:
\begin{equation}
\begin{array}{*3{>{\displaystyle}lc}p{5cm}}
 l-υ_{LR} |t|  &\geq& ε \dist(Λ_A +\boldsymbol n,Λ)-1+ε\diam(Λ_A \cup Λ)-1+2 - ε^{\prime}|\boldsymbol n| \\
 &\geq& ε |\boldsymbol n| -ε \diam(Λ_A \cup Λ) +ε\diam(Λ_A \cup Λ)- ε^{\prime}|\boldsymbol n|  \\
 &=& (ε-ε^{\prime})|\boldsymbol n| >0 \label{ineq:5.4lvtestimate}
 \end{array}
\end{equation}
where in the second line we used relation \ref{eq:distΛvecn} below.  Hence, we conclude that:
\begin{equation}
    \exp{-λ(l-υ_{LR}|t|)}  \leq \exp{-λ(ε-ε^{\prime})|\boldsymbol n|} , \, \forall t \in T, 0<ε^{\prime}<ε<1
\end{equation}
Now since $\lim_r \exp{-λ(r-υ_{LR}|t|)} =0$, $\forall t \in T$, we see that $$\lim_{\boldsymbol n \to \infty} \lim_{r \to \infty} T_2 =0.$$

We  estimate $\dist(Λ_A +\boldsymbol n, Λ_l)$ with respect to $|\boldsymbol n|$ by simple geometric arguements:
\begin{equation}
\begin{array}{*3{>{\displaystyle}lc}p{5cm}}
        \dist(Λ_A+\boldsymbol n, Λ_l) &\geq& \dist(Λ_A +\boldsymbol n, Λ)-l \\
        &\geq& (1-ε) \dist(Λ_A +\boldsymbol n, Λ) - ε \diam(Λ_A \cup Λ)-2
        \end{array}
\end{equation}
and
\begin{equation}
 \begin{array}{*3{>{\displaystyle}lc}p{5cm}}
    \dist(Λ_A+\boldsymbol n, Λ) &=& \min \{ |\boldsymbol y-\boldsymbol z| : \boldsymbol y \in Λ_A+\boldsymbol n, \boldsymbol z \in Λ \} \\
        &=& \min \{ |\boldsymbol y+\boldsymbol n-\boldsymbol z| : \boldsymbol y \in Λ_A, \boldsymbol z \in Λ \} \\
        &\geq& \min \{ |\boldsymbol n| - |\boldsymbol y-\boldsymbol z| :\boldsymbol y \in Λ_A, \boldsymbol z \in Λ \} \\
        &=& |\boldsymbol n| - \max\{ |\boldsymbol y-\boldsymbol z| : \boldsymbol y \in Λ_A ,\boldsymbol z \in Λ \} \\
        &\geq& |\boldsymbol n| - \diam(Λ_A \cup Λ)
    \end{array} \label{eq:distΛvecn}
\end{equation}
\end{proof}

\section{Von Neumann's mean ergodic theorem}
Von Neumann's mean ergodic theorem was first shown in \cite{neumann1932ergodic}:
\begin{theorem} \label{th:neumann}
   Let $U$ be a unitary operator on a Hilbert sapce $H$ and $P$ the orthogonal projection onto $ker(U-\mathds{1})$: the space invariant under $U$. Then for any $Ψ\in H$ von Neumann showed that 
    \begin{equation}
        \lim_{N} \frac{1}{N} \sum_{n=0}^N U^nΨ = PΨ
    \end{equation}
\end{theorem}
\bibliography{referencesFINAL}

\end{document}